\documentclass[11pt]{article}
\RequirePackage[numbers]{natbib}
\RequirePackage[colorlinks,citecolor=blue,urlcolor=blue]{hyperref}
\RequirePackage{graphicx}

\usepackage[utf8]{inputenc}
\usepackage{graphicx,amstext,amsmath,amsthm,xcolor,bbm,mathtools,array,titletoc,fancyhdr,listings, tcolorbox,bm,tikz}
\usepackage[ruled,linesnumbered]{algorithm2e}
\usepackage{booktabs} 
\usepackage{multirow} 
\usetikzlibrary{calc,decorations.pathreplacing,angles,quotes}
\usepackage[T1]{fontenc}
\usepackage{newtxtext}
\usepackage{newtxmath}

\theoremstyle{plain}

\newtheorem{theorem}{Theorem}
\newtheorem{lemma}[theorem]{Lemma}
\theoremstyle{definition}

\theoremstyle{remark}

\newtheorem{proposition}[theorem]{Proposition}
\newtheorem{corollary}{Corollary}

\DeclareMathOperator*{\argmax}{argmax}

\newcommand{\iid}{\stackrel{\mathrm{i.i.d.}}{\sim}}

\title{Detecting change regions on spheres}
\author{Di Su$^1$, Yining Chen$^1$, Tengyao Wang$^1$}
\date{$^1$London School of Economics and Political Science}
\begin{document}

\maketitle

\begin{abstract} 
	While change point detection in time series data has been extensively studied, little attention has been given to its generalisation to data observed on spheres or other manifolds, where changes may occur within spatially complex regions with irregular boundaries, posing significant challenges. We propose a new class of estimators, namely, \textbf{C}hange \textbf{R}egion \textbf{I}dentification and \textbf{S}e\textbf{P}aration (CRISP), to locate changes in the mean function of a signal-plus-noise model defined on $d$-dimensional spheres. The CRISP estimator applies to scenarios with a single change region, and is extended to multiple change regions via a newly developed generic scheme. The convergence rate of the CRISP estimator is shown to depend on the VC dimension of the hypothesis class that characterises the change regions in general. We also carefully study the case where change regions have the geometry of spherical caps. Simulations confirm the promising finite-sample performance of this approach. The CRISP estimator's practical applicability is further demonstrated through two real data sets on global temperature and ozone hole.
\end{abstract}

\section{Introduction}\label{sec:intro}
Detecting changes or abnormal regions in directional data arises naturally in many scientific applications~\citep{watson1983statistics,mardia2009directional, pewsey2021recent, saavedra2022nonparametric}. For instance, changes in atmospheric measurements such as ozone levels or wind direction may indicate significant environmental events~\citep{cicerone1987changes,porte2013numerical}. In such settings, we observe responses along different directions, which can be viewed as points on the unit sphere. The aim is to find regions on the sphere where the behaviour of the data inside is different. This problem falls under the general framework of change region detection.

Change detection has been studied extensively across multiple research fields. In statistics, classical change point and change region detection methods are predominantly developed for Euclidean domains. Much of the literature focuses on temporally ordered observations, such as time series, across a wide range of settings. See, for instance,  \citet{Kokoszka2000ChangepointEI}, \citet{shao2010testing}, \citet{anderson2011statistical}, \citet{Zhou2013HeteroscedasticityAA}, \citet{wang2018high}, \citet{Dette2019DetectingRC} and \citet{baranowski2019narrowest}. See also \citet{tartakovsky2014sequential} for a comprehensive introduction. When observations are not temporally ordered but instead arise from multivariate spatial designs, the problem shifts from detecting change points to detecting change regions. In Euclidean spaces, methods have been developed to identify rectangular regions of change \citep{madrid2021lattice}, as well as change regions subject to shape or smoothness constraints \citep{1459050,brunel2013adaptive,brunel2018methods,chan2022inference}. These works primarily address the theoretical limits and feasibility of detection, with comparatively less emphasis on practical implementation.
Related ideas also appear in computer vision. Foreground object detection aims to separate signal from background in image data \citep{stauffer1999adaptive,kirillov2023segment}, while point-cloud segmentation methods seek to identify structured regions in unstructured spatial data \citep{qi2017pointnet}. However, such approaches are typically designed for gridded image domains and rely on very high signal-to-noise ratios, limiting their applicability to more general spatial or directional data settings.

This work is motivated by two real datasets with directional structures. The first is based on the European Centre for Medium-Range Weather Forecasts (ECMWF) Re-Analysis v5~\citep{cucchi2020wfde5} (a.k.a. ERA5) and European Space Agency (ESA) Climate Change Initiative (CCI) global temperature data~\citep{embury2024satellite}, which provide temperature measurements over the Earth’s surface across multiple decades. The second is satellite-based ozone data from NASA's Ozone Monitoring Instrument (OMI)~\citep{levelt2018ozone}, which records total column ozone over the Southern Hemisphere. In both cases, the data lie on the sphere and exhibit spatially localised regions of change over time, for instance, warming over land masses or ozone depletion over Antarctica. Our interest is to detect such change regions based on these noisy observations on the sphere.

While change point detection in one-dimensional sequences such as time series is well-studied, detecting change regions in noisy directional data remains much less explored. The shift from detecting a single change in a sequence to detecting a spatial region of change on a directional domain introduces both modeling and computational challenges. In particular, the observations lie on a smooth manifold such as a sphere, and the change may occur within an unknown region of potentially irregular shape. As a result, implementation is typically time-consuming, with larger sample sizes than in time series studies required for accurate detection.

Our primary application is change-region detection for directional data on the sphere, but the methodology and analysis developed here apply more broadly to detecting change regions of a prescribed shape class on manifolds under a signal-plus-noise model. We begin with the simpler problem of estimating a single change region in Section~\ref{sec:subsec:singleChangeRegion}, assuming that the region belongs to a class of controlled complexity. For the multiple change region setting, and for computational tractability, we focus on the case where change regions are spherical discs in Section~\ref{sec:subsec:MultiChangeRegions}, and mention generalisation to other shape classes in Section~\ref{Sec:Extension}. The proposed estimator, which we call CRISP (\textbf{C}hange \textbf{R}egion \textbf{I}dentification and \textbf{S}e\textbf{P}aration), combines a CUSUM-type scan statistic with a local residual test designed to eliminate false positives, resulting in a fully data-driven procedure. Here the local residual test ensures that estimation is carried out within neighbourhoods containing at most one true change region, thereby serving a role analogous to Narrowest-over-Threshold procedures \citep{baranowski2019narrowest} in temporal change point analysis. We believe that this local residual testing idea is generic, and may be of independent interest for other change region or change point estimation problems. We also establish consistency and derive explicit convergence rates in both the single- and multiple-region settings (see Theorems~\ref{thm:consistency} and \ref{thm:multi_consistency}), showing that the estimation error depends on the richness of the candidate class used to model the change regions. Finally, we address the computational aspects of the proposed approach, with particular attention to the multiple-region case, where the size of the search space poses significant challenges.

The remainder of the paper is organised as follows. Section~\ref{sec:method} introduces the statistical model and describes the proposed estimators and their theoretical guarantees for both single and multiple change region settings. Section~\ref{sec:simulation} provides empirical results on simulated datasets for both single and multiple change regions with various change values and sample sizes.  Section~\ref{sec:real_data} provides the application of our proposal to real datasets including global temperature change and ozone depletion. 
Proofs and technical lemmas are deferred to the appendices.
\subsection{Notation}
Let $\mathbb{N}^{+}:=\{1,2,\dots\}$. For $r\in\mathbb{N}^{+}$, let $[r]$ denote the set $\{1,\dots,r\}$ and let $\mathcal{S}_r$ denote the set of all permutations of $[r]$.  For two points $a,b$ on a Riemannian manifold $\mathcal{M}$, we denote their geodesic distance by $\mathrm{Geo}(a,b)$. For two arbitrary sets $A,B\subseteq\mathcal{M}$, denote their symmetric difference by $A \triangle B$, denote $\mathrm{dist}(A,B):=\min_{a\in A,b\in B}\mathrm{Geo}(a,b)$, and their Hausdorff distance as $d_H(A,B):=\max\{\sup_{x\in A}\mathrm{dist}(\{x\},B),\sup_{y\in B}\mathrm{dist}(\{y\},A)\}$. For $\epsilon > 0$, we write $\mathrm{Nhd}(A,\epsilon):=\{x\in\mathcal{M}:\mathrm{dist}(\{x\},A)\leq\epsilon\}$. Let $\mathbb{S}^{d-1}:=\{x\in\mathbb{R}^{d}:\lVert x\rVert=1\}$ be the $d$-dimensional sphere and for $\alpha\in\mathbb{S}^{d-1}$ and $\beta\in[0,1]$, let $A_{\alpha,\beta}=\{x\in\mathcal{M}:x^{T}\alpha\geq\beta\}$ denote a $d$-dimensional disc on the sphere. We write $\mathcal{S}=\{A_{\alpha,\beta}\subseteq\mathbb{S}^{d-1}:\alpha\in\mathbb{S}^{d-1},\beta\in[0,1]\}$ for the collection of discs on $\mathbb{S}^{d-1}$. For a disc $A\in\mathcal{S}$, let $\operatorname{Rad}(A)$ and $\operatorname{Ctr}(A)$ denote the radius and center of $A$ respectively.  We write $\mathbbm{1}_{E}$ for the indicator of an event $E$.

\section{A general framework for change region detection on manifolds}
\label{sec:method}
In this section, we describe our CRISP methodology and provide guarantees on its estimation accuracy.
We start with a general compact Riemannian manifold, and specialise to the sphere $\mathbb{S}^{d-1}$ later.

Let $X_1,\dots,X_n$ be $n$ design points on a Riemannian manifold $\mathcal{M}$, which are either deterministic or drawn as an independent and identically distributed sample from a probability measure $P_X$ on $\mathcal{M}$. 

Assume that for some $r\in\mathbb{N}^{+}$, there are $r$ disjoint change regions $R_1, \dots, R_r\subseteq\mathcal{M}$, and we have observations $(X_i,Y_i),i=1,
\dots,n$ such that
\begin{align}
  \label{eqn:dataset_multi}
  Y_i=\sum_{j=1}^{r}\mu^{(j)}\mathbbm{1}\{X_i\in R_j\}+\mu^{(0)}\mathbbm{1}\{X_i\notin \cup_{j\in[r]} R_j\}+\varepsilon_i,\quad i\in[n],
\end{align}
where $\varepsilon_i \iid N(0,\sigma^2)$ and $\mu^{(j)}\in\mathbb{R}$ for $j = 0$ and every  $j \in [r]$.
Let
\begin{equation}
  \label{Eq:Thetaj}
  \theta_j:=|\mu^{(j)}-\mu^{(0)}|,j\in[r],
\end{equation}
denote the difference between the means of the observations inside the region $R_j,j\in\{1,\dots,r\}$ and outside any change region. In particular, when $r=1$ and there is a single change region, the data generating process is
\begin{align}
  \label{eqn:dataset}
  Y_i=\mu^{(1)}\mathbbm{1}\{X_i\in R_1\}+\mu^{(0)}\mathbbm{1}\{X_i\notin R_1\}+\varepsilon_i,\quad i\in[n],
\end{align}
and we denote
\begin{equation}\label{Eq:Theta}
  \theta=|\mu^{(1)}-\mu^{(0)}|
\end{equation}

Change region detection is the task of estimating the sets $\mathcal{R}:=\{R_1,\dots,R_r\}$ from the data $\{(X_i,Y_i):i\in[n]\}$. We write $\mathcal{D}=\{X_1,\dots,X_n\}$ and $\mathcal{Y}:=\{Y_1,\dots,Y_n\}$, and $|A|_{\mathcal{D}} := |A \cap\mathcal{D}|$ for any set $A\subseteq \mathcal{M}$. For a fixed $j\in\{1,\dots,r\}$, the quality of any estimator $\hat{R}_j$ of $R_j$ can be measured in terms of in-sample classification error
\begin{equation}
  \label{eqn:def:loss}
  L_n(\hat{R}_j, R_j) := \frac{1}{n}\min\bigl(|\hat{R}_j\triangle R_j|_{\mathcal{D}},|\hat{R}_j\triangle R_j^{\text{c}}|_{\mathcal{D}}\bigr).
\end{equation}
Here, the loss is defined as the minimum of the empirical measure $n^{-1}\sum_{i=1}^n \delta_{X_i}$ evaluated on $\hat{R}_j\triangle R_j$ and $\hat{R}_j\triangle R_j^{\mathrm{c}}$, because in the above single-change-region model, the change region $R_j$ is only identifiable up to taking set complement. Alternatively, if $X_i$ are sampled from a distribution $P_X$, we can measure the loss of $\hat R_j$ in terms of its generalisation error
\begin{equation}
\label{Eq:PopLoss}
  L(\hat R_j, R_j):= \min\{P_X(\hat{R}_j\triangle R_j) , P_X(\hat{R}_j\triangle R_j^{\mathrm{c}})\}.    
\end{equation}

The general problem of estimating $\mathcal{R}$ given data $(X_i,Y_i)$, $i\in[n]$ is not possible without imposing further structural assumptions on the class of sets $\mathcal{R}$. We assume here that $R_1,\dots,R_r$ belong to a family of possible change regions $\mathcal{A}$. For instance, if we are interested to detect solar flares in a circular region on the surface of the sun, we would take $\mathcal{M}$ to be the Euclidean sphere $\mathbb{S}^2$ in $\mathbb{R}^3$ and $\mathcal{A}$ to be the set of all possible discs on $\mathbb{S}^2$.

Such structural assumptions restrict the complexity of the candidate class of change regions, which in turn governs the convergence rate of the estimator. For a family of possible change regions $\mathcal{A}$, we use its  Vapnik–Chervonenkis (VC) dimension \citep{vapnik2015uniform}, denoted by $\mathrm{VCD}(\mathcal{A})$, to quantify its complexity. The VC dimension is a classical measure of complexity that captures how well a class of sets can distinguish between different configurations of points, which makes it naturally suited to our task of separating data points into change and no-change regions. A higher VC dimension implies a more expressive class, but also one that may be harder to estimate accurately from data. In this work, we assume that $\mathcal{A}$ has finite VC dimension.

\subsection{Single change region}\label{sec:subsec:singleChangeRegion}
We first consider the case where there is exactly one change region $R_1$ on $\mathcal{M}$ (i.e., $r=1$). Given that the change region belongs to a family $\mathcal{A}$ of subsets of the manifold $\mathcal{M}$, we construct the change region estimator by maximising corresponding cumulative sum (CUSUM) contrast statistics. Here, for any candidate change region $A$, the CUSUM statistic of $\{(X_i, Z_i):i\in[n]\}$ at the region $A$ is defined as
\begin{align}
  \label{Eq:CUSUM}
  \mathcal{T}_{A}((X_i,Z_i)_{i\in[n]}) = \sqrt\frac{{|A|_{\mathcal{D}}|A^{\text{c}}|_{\mathcal{D}}}}{n}\biggl\{\frac{1}{|A|_{\mathcal{D}}}\sum_{i:X_i\in A}Z_i-\frac{1}{|A^{\text{c}}|_{\mathcal{D}}}\sum_{i:X_i\in A^{\text{c}}}Z_i\biggr\},
\end{align}
with the convention that $\mathcal{T}_A((X_i,Z_i)_{i\in[n]}) = 0$ if either $|A|_{\mathcal{D}}=0$ or $|A^{\mathrm{c}}|_{\mathcal{D}} = 0$.
When the design points $(X_i)_{i\in[n]}$ is clear from the context, we will abbreviate $\mathcal{T}_A(Z)=\mathcal{T}_{A}((X_i,Z_i)_{i\in[n]})$ for $Z=(Z_1,\dots,Z_n)^{\top}$.\par
Writing $Y = (Y_1,\ldots,Y_n)^\top$, we estimate $R_1$ by
\begin{align}\label{eqn:def:estimator}
  \hat{R} \in \argmax_{A\in\mathcal{A}} \, |\mathcal{T}_A(Y)|.
\end{align}
This estimator is intuitive because $R_1$ maximises the noiseless CUSUM $\mathcal{T}_A(\mu)$ for $\mu:=(\mu_1,\dots,\mu_n)^{\top}$ where $\mu_i:=\mu_{1}\mathbbm{1}(X_i\in R_1)+\mu_{0}\mathbbm{1}(X_i\notin R_1)$. By linearity, we have $\mathcal{T}_A(Y) = \mathcal{T}_A(\mu) + \mathcal{T}_A(\varepsilon)$, where $\varepsilon := (\varepsilon_1,\ldots,\varepsilon_n)^\top$. Thus, treating $\mathcal{T}_A(Y)$ as a perturbation of $\mathcal{T}_A(\mu)$, we expect $\hat R$ to be close to $R_1$.

For any $A\subseteq \mathcal{M}$ and any vector $Z = (Z_1,\ldots,Z_n)^\top$, $\mathcal{T}_A(Z)$ depends on $A$ only through $A\cap \mathcal{D}$. Moreover, as mentioned earlier, the change-region model in~\eqref{eqn:dataset} only identifies $R_1$ up to set complements, and the CUSUM statistic in~\eqref{Eq:CUSUM} is also invariant to replacing $A$ with $A^{\mathrm{c}}$. Therefore, given design points $(X_i)_{i\in[n]}$, both $R_1$ and $\hat R$ are only identifiable up to their intersection with $\mathcal{D}$ and up to taking set complements. Such identifiability issues motivate our choice of the loss function $L_n$ in~\eqref{eqn:def:loss} since $L_n(A_1,A_2) = L_n(A_1,A_2^{\text{c}})=L_n(A_1^{\text{c}},A_2^{\text{c}})$ for any $A_1,A_2\in\mathcal{A}$.

The following theorem establishes the consistency of the estimated change region described above.

\begin{theorem}
  \label{thm:consistency}
  Assume $r=1$, and $R_1\in\mathcal{A}$ where the family of possible change regions $\mathcal{A}$ is a VC class. Given distinct fixed design points $X_1,\ldots, X_n$, let $Y_1,\ldots, Y_n$  be generated according to~\eqref{eqn:dataset_multi} and let $\theta$ be the magnitude of change defined as in~\eqref{Eq:Theta}. If $\tau:= L_n(R_1,\emptyset) >0$, then there exist universal constants $C_0, C_1,C_2$ such that when $n\tau\theta^2 > C_0\sigma^2 \mathrm{VCD}(\mathcal{A})\log n$, with probability at least $1 - C_1 n^{-\mathrm{VCD}(\mathcal{A})}$, we have that the estimator $\hat{R}$ defined in~\eqref{eqn:def:estimator} satisfies
  \begin{align*}
    L_n\bigl(\hat{R},R_1\bigr)\leq \frac{C_2\sigma^2\mathrm{VCD}(\mathcal{A})\log(n)}{n\tau \theta^2}.
  \end{align*}
\end{theorem}
Theorem~\ref{thm:consistency} highlights how the estimation error depends on three key factors: the complexity of the candidate class, the identifiability of the true change region, and the signal-to-noise ratio. Specifically, the error increases with the richness of the class $\mathcal{A}$ quantified by its VC dimension, and decreases with stronger signal $\theta$, larger sample size $n$, and better empirical identifiability, measured by $\tau = L_n(R_1, \emptyset)$.

In this result, the parameter $\tau$ captures the smaller of the proportions of sample points inside and outside the change region, and reflects how well-separated the region is in terms of data coverage. When $\tau$ is close to zero, the change region is either too small or too poorly sampled to be reliably distinguished from the background. The condition $n \tau \theta^2 >C_0\sigma^2 \mathrm{VCD}(\mathcal{A}) \log n$ ensures that the signal-to-noise ratio, adjusted for the effective sample size within the region, is large enough to detect the change reliably. This is analogous to minimal spacing or minimal signal assumptions in classical change point detection. In practice, $\tau$ reflects the fact that one cannot recover regions that are too small relative to the total sample, even with a strong signal. Hence, the theorem highlights the interaction between region size $\tau$, signal strength $\theta$, and model complexity $\mathrm{VCD}(\mathcal{A})$ in determining the feasibility of consistent estimation.

The convergence rate in Theorem~\ref{thm:consistency} simplifies to $O(1/n)$ up to logarithmic and complexity terms under fixed $\tau$ and $\theta$. This is sometimes referred to as a superparametric rate, as it is faster than the standard $n^{-1/2}$ rate encountered in regular parametric estimation problems. Such rates are typical in support recovery problems, including classical change point detection, where the object of interest is a set (e.g., the location of a jump) rather than a smooth parameter. The reason for this faster rate is that the signal manifests as a mean shift over a subset of the domain, and the estimator aggregates evidence over that subset, reducing variance effectively. In our setting, the estimator recovers the change region by maximizing a CUSUM-type statistic, and the loss function based on symmetric difference reflects a discrete support estimation task. Similar $1/n$ rates have been observed in univariate change point literature~\citep{padilla2022change}, and our result applies to more complex geometric settings on manifolds.

We specialise the above result to problem of detecting circular regions of change for directional data in the following corollary. More precisely, we assume that the design points lie on the unit sphere $\mathbb{S}^{d-1}$ and the change region belongs to the class of spherical discs $\mathcal{S}$.
\begin{corollary}
  \label{Cor:EstimationDiscs}
  Assume $r=1$ and $R_1=\{x\in\mathbb{S}^{d-1}:x^T\alpha\geq\beta\}$ for some constants $\alpha\in\mathbb{S}^{d-1}$ and $\beta\in[0,1]$. Let $(X_i, Y_i),\ldots, (X_n, Y_n)$ be generated according to~\eqref{eqn:dataset_multi} and $\theta$ be the magnitude of change defined as in~\eqref{Eq:Theta}. Let $\hat{R}=\{x\in\mathbb{S}^{d-1}:x^T\hat \alpha\geq\hat \beta\}$ be the estimator as defined in~\eqref{eqn:def:estimator} with $\mathcal{A}=\mathcal{S}$. Assume $L_n(R_1,\emptyset) = \tau>0$. There exist universal constants $C_1,C_2$ such that for $n$ large enough, with probability at least $1-C_1n^{-(d+1)}$, the change region estimation has a rate of convergence
  \begin{align}\label{Eqn:CoroSingleCPRegionEst}
    L_n\bigl(\hat{R},R_1\bigr)\leq \frac{C_2\sigma^2d\log(n)}{\theta^2\tau n}.
  \end{align}
  Moreover, 
  if the design points $X_1,\ldots,X_n \iid \mathrm{Unif}(\mathcal{M})$, then there exists $C_{d,\beta}$ depending only on $d$ and $\beta$, such that 
  \begin{align}\label{Eqn:CoroSingleCPParaEst}
    \|\hat{\alpha}-\alpha\|+|\hat{\beta}-\beta|\leqslant C_{d,\beta}\bigg(\frac{C_2 \sigma^2 d \log (n)}{\theta^2 \tau n}+\sqrt{\frac{d\log(n)}{n}}\bigg).
  \end{align}
\end{corollary}

In addition to the result for change region estimation measured in terms of the empirical loss, Corollary~\ref{Cor:EstimationDiscs} also translates the result to parameter estimation of the change region's center and radius.
Compared to the $O(1/n)$ rate in~\eqref{Eqn:CoroSingleCPRegionEst} for region estimation, the parameter estimation rate in Corollary~\ref{Cor:EstimationDiscs} contains an additional $\sqrt{\log(n)/n}$ term in~\eqref{Eqn:CoroSingleCPParaEst}. This slower rate stems from an inherent identifiability issue in the model, that is, the mapping from parameters (e.g., center and radius of the disc) to induced region is many-to-one with respect to the observed data. Small perturbations in the parameters may yield regions that are identical or nearly indistinguishable in terms of their intersection with the sample points. As a result, a low in-sample classification error does not necessarily imply accurate parameter recovery. The assumption that the random design points $X_1,\ldots, X_n$ come from a uniform distribution on the compact manifold can be weakened. The same result will hold if we replace the uniform distribution by any distribution $P_X$ whose density on $\mathcal{M}$ is lower bounded by $\underline{f}_X$. The corresponding result will have the leading constant depending on $\underline{f}_X$ as well.

The matching minimax lower bound in Proposition~\ref{Prop:LowerBound} in Appendix~\ref{appB} shows that the convergence rate in Corollary~\ref{Cor:EstimationDiscs} is essentially optimal. Crucially, the lower bound scales linearly with the VC dimension, illustrating that the richness of the region class not only affects the difficulty of estimation as seen in the upper bound but also sets a fundamental limit on performance. This result shows that no estimator can achieve a faster rate up to logarithmic terms over the same model class. 
\subsection{Multiple change regions}\label{sec:subsec:MultiChangeRegions}
In this subsection, we focus on the multiple change region setting with $\mathcal{M} = \mathbb{S}^{d-1}$ and $\mathcal{A} = \mathcal{S}$.
When $r>1$ and multiple change regions are present, change region estimation becomes substantially more challenging. The key difficulty is that the CUSUM contrast designed for a single change region is no longer well aligned with the problem: a global contrast for a given region $R_i$ typically compares observations drawn from several distinct regimes, because other change regions contaminate the background set. As a result, signal contributions may cancel, resulting in difficulties in detecting some change regions.

To mitigate this signal attenuation, we compute contrast statistics locally. Specifically, we construct a collection of local neighbourhood discs
\[
  \mathcal{B}_n=\{B_1,\ldots,B_{J_d}\}\subseteq\mathcal{S},
\]
randomly sampled such that $(\mathrm{Ctr}(B_j), \mathrm{Rad}(B_j)) \iid \mathrm{Unif}(\mathbb{S}^{d-1})\otimes \mathrm{Unif}[0,\pi]$. Within each disc $B\in\mathcal{B}_n$ we search for an inner disc that captures a potential change region while remaining well separated from the boundary of $B$. For a fixed $\omega\in(0,2\pi)$, we restrict attention to discs $A\subseteq B$ satisfying $\mathrm{dist}(A,B)\ge \omega/2$, and define the local CUSUM statistic
\begin{equation}
  \label{Eq:localCUSUM}
  \mathcal{T}_{A}^{B}(Z)
  =\sqrt{\frac{|A|_{\mathcal{D}}\,|A^{\mathrm{c}}|_{\mathcal{D}}}{|B|_{\mathcal{D}}}}
  \left|
  \frac{\sum_{i:X_i\in A}Z_i}{|A|_{\mathcal{D}}}
  -\frac{\sum_{i:X_i\in B\setminus A}Z_i}{|B\setminus A|_{\mathcal{D}}}
  \right|,
\end{equation}
with the convention that $\mathcal{T}^B_A(Z)=0$ if $|A|_{\mathcal{D}}=0$ or $|B\setminus A|_{\mathcal{D}}=0$. For each $B$, we define
\[
  \hat R_B \in \argmax_{A\subseteq B:\,\mathrm{dist}(A,B^{\mathrm{c}})\ge \omega/2}
  |\mathcal{T}^{B}_A(Y)|,
\]
and regard $\hat R_B$ as a candidate change region detected within $B$.

A large local CUSUM statistic alone, however, does not ensure that $\hat R_B$ is a good estimate of a single, true change region. In particular, if $B$ contains more than one true change region, an inner disc $A$ that aggregates parts of several regions may produce an even larger contrast, despite being poorly aligned with any individual region; see Figure~\ref{fig:cusum_BmanyR_BpartialR}(a) for an example configuration.

Ideas such as Narrowest-Over-Threshold (NOT) \citep{baranowski2019narrowest}, which are designed to handle similar situations for multiple change point estimation problems, are not well suited to the present setting. In particular, NOT-type procedures cannot exclude cases where a small local region $B$ captures only a fragment of a true change region $A$, yet still produces a significant contrast statistic; see Figure~\ref{fig:cusum_BmanyR_BpartialR}(b) for an example configuration. Instead, to rule out such spurious candidates, we introduce a secondary residual-based check. For any $A\subseteq B$, define the residual sum of squares
\begin{equation}
  \label{Eq:RSS}
  \mathrm{RSS}^{B}_{A}(Z)
  =\sum_{i:X_i\in A}\left(Z_i-\frac{\sum_{i:X_i\in A}Z_i}{|A|_{\mathcal{D}}}\right)^2
  +\sum_{i:X_i\in B\setminus A}\left(Z_i-\frac{\sum_{i:X_i\in B\setminus A}Z_i}{|B\setminus A|_{\mathcal{D}}}\right)^2,
\end{equation}
with the usual empty-set conventions. Intuitively, $\mathrm{RSS}^{B}_{A}(Y)$ is small only when both $A$ and $B\setminus A$ are approximately homogeneous, each well described by a single mean. This property fails when either set mixes observations from multiple regimes, which is precisely the situation that can inflate the local CUSUM in the multi-region setting.  This residual-based screening therefore provides an effective safeguard against fragmented or aggregated detections and may be of independent interest for other change region and change point problems.

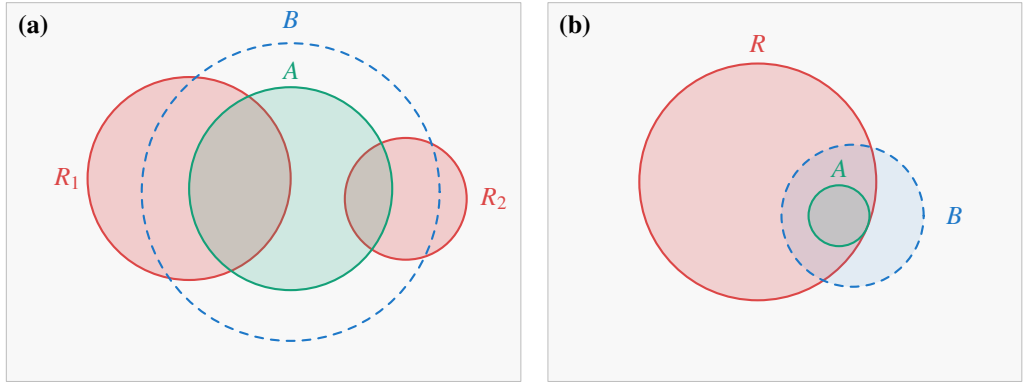
\begin{figure}
  \centering
  \begin{tikzpicture}[scale=0.9, transform shape, line cap=round, line join=round]
  \definecolor{bg}{RGB}{248,248,248}
  \definecolor{region}{RGB}{220,70,70}
  \definecolor{scan}{RGB}{30,120,200}
  \definecolor{agg}{RGB}{20,160,120}

  \def\W{8.4}
  \def\gap{-0.6}
   \begin{scope}[shift={(0,0)}]
    \fill[bg] (-4.2,-2.8) rectangle (3.4,2.8);
    \draw[black!25] (-4.2,-2.8) rectangle (3.4,2.8);
    \fill[region, opacity=0.25] (-1.5,0.2) circle (1.5);
    \draw[region, thick] (-1.5,0.2) circle (1.5);
    \node[region] at (-3.3,0.2) {$R_1$};

    \fill[region, opacity=0.25] (1.7,-0.1) circle (0.9);
    \draw[region, thick] (1.7,-0.1) circle (0.9);
    \node[region] at (3,-0.1) {$R_2$};

    \draw[scan, thick, dashed] (0,0) circle (2.2);
    \node[scan] at (0,2.55) {$B$};

    \fill[agg, opacity=0.18] (0,0.05) circle (1.5);
    \draw[agg, thick] (0,0.05) circle (1.5);
    \node[agg] at (0,1.8) {$A$};

    \node[black, font=\bfseries] at (-3.8,2.45) {(a)};
  \end{scope}
  \begin{scope}[shift={(\W+\gap,0)}]
    \fill[bg] (-4,-2.8) rectangle (3,2.8);
    \draw[black!25] (-4,-2.8) rectangle (3,2.8);

    \fill[region, opacity=0.22] (-0.9,0.15) circle (1.75);
    \draw[region, thick] (-0.9,0.15) circle (1.75);
    \node[region] at (-0.9,2.2) {$R$};

    \fill[scan, opacity=0.10] (0.5,-0.35) circle (1.05);
    \draw[scan, thick, dashed] (0.5,-0.35) circle (1.05);
    \node[scan] at (2,-0.35) {$B$};

    \fill[agg, opacity=0.12] (0.3,-0.35) circle (0.45);
    \draw[agg, thick] (0.3,-0.35) circle (0.45);
    \node[agg] at (0.3,0.35) {$A$};

    \node[black, font=\bfseries] at (-3.6,2.45) {(b)};
  \end{scope}
\end{tikzpicture}
 \caption{(a) A scan region $B$ contains multiple true change regions; an inner disc $A$ can aggregate parts of several regions. (b) A small scan region $B$ captures only a fragment of a true change region $R$, yet can still yield a significant contrast.}
  \label{fig:cusum_BmanyR_BpartialR}
\end{figure}

We therefore retain $\hat R_B$ as an initial candidate only if it satisfies both
\[
  \mathcal{T}^{B}_{\hat R_B}(Y)\ge \lambda_T
  \quad\text{and}\quad
  \mathrm{RSS}^{B}_{\hat R_B}(Y)\le \gamma_B,
\]
where the threshold $\lambda_T$ is derived from the null behaviour of the CUSUM statistics, and $\gamma_B$ is a suitable RSS threshold. Here the threshold $\gamma_B$ depends on the degrees of freedom, $m$, of the residual sums of squares computed. For instance, we can take
\[
  \gamma_B = m + 2\sqrt{m \lambda_R} + 2\lambda_R,
\]
for a suitable choice of $\lambda_R > 0$, following tail bounds of chi-squared distributions \citep{laurent2000adaptive}.

Applying this procedure to all discs in $\mathcal{B}_n$ yields an initial collection of candidate regions $\hat{\mathcal{R}}_{\mathrm{init}}$. Since multiple discs may detect the same underlying change region, we perform a final refinement step: we iteratively retain the candidate with the largest local CUSUM statistic and discard all remaining candidates that intersect it, repeating until all retained regions are pairwise disjoint. The resulting set $\hat{\mathcal{R}}$ forms our final estimate of the change regions. The complete description of the CRISP procedure is summarised in Algorithm~\ref{algo:multiple_cp_detection}, with the consistency of the estimators from this algorithm established below in Theorem~\ref{thm:multi_consistency}.

\begin{algorithm}[htbp]
  \caption{CRISP: Pseudo-code for the estimation of multiple change regions}
  \label{algo:multiple_cp_detection}
  \KwIn{Data $(X_1,Y_1),\ldots,(X_n,Y_n)\in\mathbb{S}^{d-1}\times \mathbb{R}$; $J_d \in\mathbb{N}$; $\omega > 0$; $\lambda_{T}, \lambda_{R} > 0$}
  \KwOut{A set of estimated change regions $\hat{\mathcal{R}}$}

  Draw independent and identically distributed $B_1,\ldots,B_{J_d}$  such that $\bigl(\mathrm{Ctr}(B_1), \mathrm{Rad}(B_1)\bigr)\sim \mathrm{Unif}(\mathbb{S}^{d-1})\otimes \mathrm{Unif}[0,\pi]$.

  Initialise $\hat{\mathcal{R}}_{\mathrm{init}}\gets\emptyset$\\

  \For{$j\in\{1,\dots,J_d\}$}{
    Compute $\hat R_j\gets\argmax_{A\in\mathcal{S}:A\subseteq B_j,\mathrm{dist}(A,B_j^c)\geq\omega/2}\mathcal{T}_A^{B_j}(Y)$\\
    Compute $\mathcal{T}_j(Y) := \mathcal{T}^{B_j}_{\hat R_j}(Y)$ (cf.~\eqref{Eq:localCUSUM})\\
    Compute $\mathrm{RSS}_{j}(Y): = \mathrm{RSS}^{B_j}_{\hat R_j}(Y)$ (cf.~\eqref{Eq:RSS})\\
    Set $m_j := |B_j|_{\mathcal{D}} - 2$ and $\gamma_j:= m_j + 2\sqrt{m_j\lambda_R}+2\lambda_R$ \\
    \If{$\mathrm{RSS}_{j}(Y)<\gamma_j$ \textup{and} $\mathcal{T}_{j}(Y)>\lambda_{T}$}{
      $\hat{\mathcal{R}}_{\mathrm{init}}\gets \hat{\mathcal{R}}_{\mathrm{init}}\cup \{\hat R_j\}$
    }
  }
  Let $\hat{R}_1,\ldots, \hat{R}_{m}$ be the elements from $\hat{\mathcal{R}}_{\mathrm{init}}$ ordered according to the values of their local CUSUM statistics in decreasing order, where $m:=|\hat{\mathcal{R}}_{\mathrm{init}}|$\\

  Set $\hat{\mathcal{R}}\gets \hat{\mathcal{R}}_{\mathrm{init}}$\\

  \For{$j=1,\ldots,m$}{
    \If{$\hat{R}_j\in\hat{\mathcal{R}}$}{
      $\hat{\mathcal{R}} \gets \{\hat R\in\hat{\mathcal{R}}: \hat R \cap \hat R_j = \emptyset\} \cup \{\hat R_j\}$
    }
  }
\end{algorithm}
\newpage
\begin{theorem}\label{thm:multi_consistency}
  For $r\in\mathbb{N}$, let $R_1,\ldots,R_r$ be change regions satisfying $\mathrm{Rad}(R_k) \in [\delta, \pi/2]$ for some $\delta > 0$ and all $k\in[r]$. Given fixed distinct design points $X_1,\ldots,X_n$, let $Y_1,\ldots,Y_n$ be generated according to~\eqref{eqn:dataset_multi} with $\theta_k$ denoting the magnitude of change in each region as defined in~\eqref{Eq:Thetaj}. Define $\theta:=\min_{k\in[r]}|\theta_k|$. For $\lambda:=4(d+1)\log n$, let $\{\hat{R}_1,\dots,\hat R_{\hat r}\}$ be the output of Algorithm~\ref{algo:multiple_cp_detection} with input $(X_i,Y_i)_{i\in[n]}$, $J_d$, $\omega$, $\lambda_T=4(n\lambda)^{1/4}+8\sqrt{\lambda}$ and $\lambda_R=\lambda$. If $\delta\geq \omega \geq C_1(\lambda/n)^{1/(2d)}$, $\min_{k, k'\in[r]} \mathrm{dist}(R_k, R_{k'})\geq \omega$ and $\omega^{d}\theta^2\geq C_2\sqrt{\lambda/n}$, for $C_1, C_2 > 0$, then there exist constants $C>0$ and $C_d>0$ depending only on $d$, such that for $n$ large enough, we have
  \begin{align*}
    \mathbb{P}\biggl\{\hat{r}=r, \;  \min_{\pi\in\mathcal{S}_r} \max_{k\in[r]} L_n(R_k,\hat{R}_{\pi(k)})\leq\frac{12\sqrt{n\lambda}+32\lambda}{n\theta^2}\biggr\} \geq 1 - re^{-C_dJ_d\omega^d} - C n^{-2d-2}.
  \end{align*}
\end{theorem}
Theorem~\ref{thm:multi_consistency} establishes a simultaneous consistency result for multiple change region estimation. It shows that, with high probability, the proposed procedure correctly recovers the number of change regions and estimates the location of each region up to a small error. The correspondence between estimated and true regions is formalised via a permutation $\pi$, ensuring that regions are compared modulo ordering. Under this matching, the theorem guarantees that the maximal estimation error across all $r$ regions decays at rate approximately $\sqrt{n\lambda}/(n\theta^2)$.

The theorem requires each change region to satisfy a minimal size condition $\delta$, which ensures identifiability and prevents regions of vanishing volume, analogous to the role of $\tau$ in the single-region setting. In addition, distinct change regions are assumed to be separated by at least a geodesic distance $\omega$, so that each region can be isolated within some local window $B_j$ considered by the algorithm. The interplay between signal strength and spatial separation is captured by the condition $\omega^{d}\theta^2 \ge C_2\sqrt{\lambda/n}$. Here, $\omega^{d}$ represents the volume of the separating gap in $d$ dimensions, and the product $\omega^{d}\theta^2$ quantifies the effective signal energy available to distinguish neighboring regions. This quantity must dominate the noise level on the right-hand side, which decreases with sample size, and is closely related to energy-type criteria that appear in the change point detection literature \citep[e.g.][]{verzelen2023optimal}.

\subsection{Extension}
\label{Sec:Extension}
While our theoretical results in the multiple region setting focus on spherical discs, the overall framework naturally extends to more general classes of shapes. The key requirement is that the candidate regions form a class with finite VC dimension, which controls the statistical complexity of the search space. Many practically relevant families such as spherical caps, ellipsoidal patches, or unions of convex sets also satisfy this condition, and the core methodology remains applicable.

In principle, the theory can be extended to these settings by analyzing the corresponding VC dimension. However, moving beyond discs introduces additional challenges, particularly in computation, since discs offer a natural parameterisation and efficient covering arguments. Nonetheless, the generalisation illustrates that the proposed approach is not tied to any specific shape, but rather to a structurally controlled family of candidate regions, making it adaptable to a wide range of manifold-based detection problems.
\section{Empirical studies}\label{sec:simulation}
\subsection{Implementation details}

We choose the number of outer discs $J_d \approx 1/\{1 - F_{\mathrm{Beta}}(\cos^2(\alpha); 1/2, d/2)\}$, where $F_{\mathrm{Beta}}(\cos^2(\alpha); 1/2, d/2)$ is the distribution function of a $\mathrm{Beta}(1/2, d/2)$ random variable. Since, for any randomly drawn direction $v$ on $\mathbb{S}^{d-1}$ and a fixed direction $v^\star$, we have $\cos^2(\angle(v,v^\star))\sim \mathrm{Beta}(1/2,d/2)$, this ensures that  with high probability, at least one outer disc falls within angular distance $\alpha=0.05$ of a true centre. The threshold parameters $\lambda_T$ and $\gamma_j$ for $j\in[J_d]$ should be chosen according to the tail distribution of the noise. In the simulations, we set $\lambda_T = \sqrt{2\log(nJ_d/2)}$ and $\gamma_j$ to be the upper 0.05 quantile of a $\chi^2_{|B_j|_{\mathcal{D}}-2}$ distribution.

To speed up computation, we restrict the inner and outer discs $A$ and $B$ to be concentric. This enables fast evaluation of the CUSUM statistics by reducing the problem to a univariate change point estimation task, obtained by projecting the data onto the direction determined by the common centre. One drawback of the above simplified inner disc choice is that the estimated change region might not be very accurate. To remedy this, we include a post-processing step, where we randomly perturb the estimated change region $\hat R$ and refine it by selecting the perturbed candidate within the outer disc $B$ that yields the largest CUSUM statistic. Additionally, we only select inner disc $A$ that is at most $2/3$ in radius of the outer disc $B$. This is to guard against scenarios where $A$ captures multiple true regions of change as illustrated in Figure~\ref{fig:cusum_BmanyR_BpartialR}(a) in the finite sample setting.

Finally, when the ambient dimension $d=2$, to avoid label switching ambiguity (since both change regions and non-change regions on $\mathbb{S}^1$ are simply arc segments), we assume that we know the direction of change, which we use to construct one-sided CUSUM test statistics.

\subsection{Empirical performance of CRISP}
\label{Sec:EmpPerf}
In this section, we perform simulation studies where we apply CRISP for both the tasks of single and multiple change region detection for $d$-dimensional directional data for $d\in\{2,3,4\}$. We set the strength of the change signal $\theta\in\{1,1.5,2,2.5,3\}$ and the sample size $n\in\{200,400,600,800,1000\}$. For multiple regions detection, we set the number of change regions $r=4$. For each of these settings, the noise $\{\varepsilon_i\}$ are independent and identically sampled from the standard normal distribution, and the centres and radii of the change regions are respectively
\begin{itemize}
  \item $r=1$, $d\in\{2,3,4\}$:
  
  Centre: $ {(1, 1, \ldots, 1)^\top}/{\sqrt{d}}$. Radius: $\arccos\!\left({3}/{4}\right) $.
  \item $r=4$, $d=2$:
  
  Centres: $(1,0)^\top,\ (0,1)^\top,\ (-1,0)^\top,\ (0,-1)^\top$. Radius: $0.3$.
  \item $r=4$, $d=3$:
  
  Centres: 
  \[
    \frac{(1,1,1)^\top}{\sqrt{3}},\ 
    \frac{(1,-1,-1)^\top}{\sqrt{3}},\ 
    \frac{(-1,1,-1)^\top}{\sqrt{3}},\ 
    \frac{(-1,-1,1)^\top}{\sqrt{3}}.
  \]
  \item $r=4$, $d=4$:
  
  Centres: \[
    \frac{(1,1,1,-1)^\top}{2},\ 
    \frac{(1,-1,-1,-1)^\top}{2},\ 
    \frac{(-1,1,-1,-1)^\top}{2},\ 
    \frac{(-1,-1,1,-1)^\top}{2}.
  \] Radius: $0.7$.
\end{itemize}

The centres are chosen to have maximal pairwise separation on the sphere, and we increase the radius as the ambient dimension $d$ increases so that the fraction of design points falling into the change regions stay roughly the same.

Figure \ref{fig:rate_singlecp} shows the empirical losses of CRISP averaged over $100$ Monte Carlo repetition when $r=1$, and the empirical losses and numbers of estimated regions when $r=4$ are shown in Figure \ref{fig:rate_multicp}. Both figures are plotted at a log scale of the loss. In the log-log plot in Figure~\ref{fig:rate_singlecp} for single change detection, the results exhibit the expected slopes with respect to the sample size $n$ (approximately $-1$) and the signal strength $\theta$ (approximately $-2$), consistent with Theorem~\ref{thm:consistency}.
\begin{figure}[htbp]
  \centering
  \includegraphics[scale=0.6]{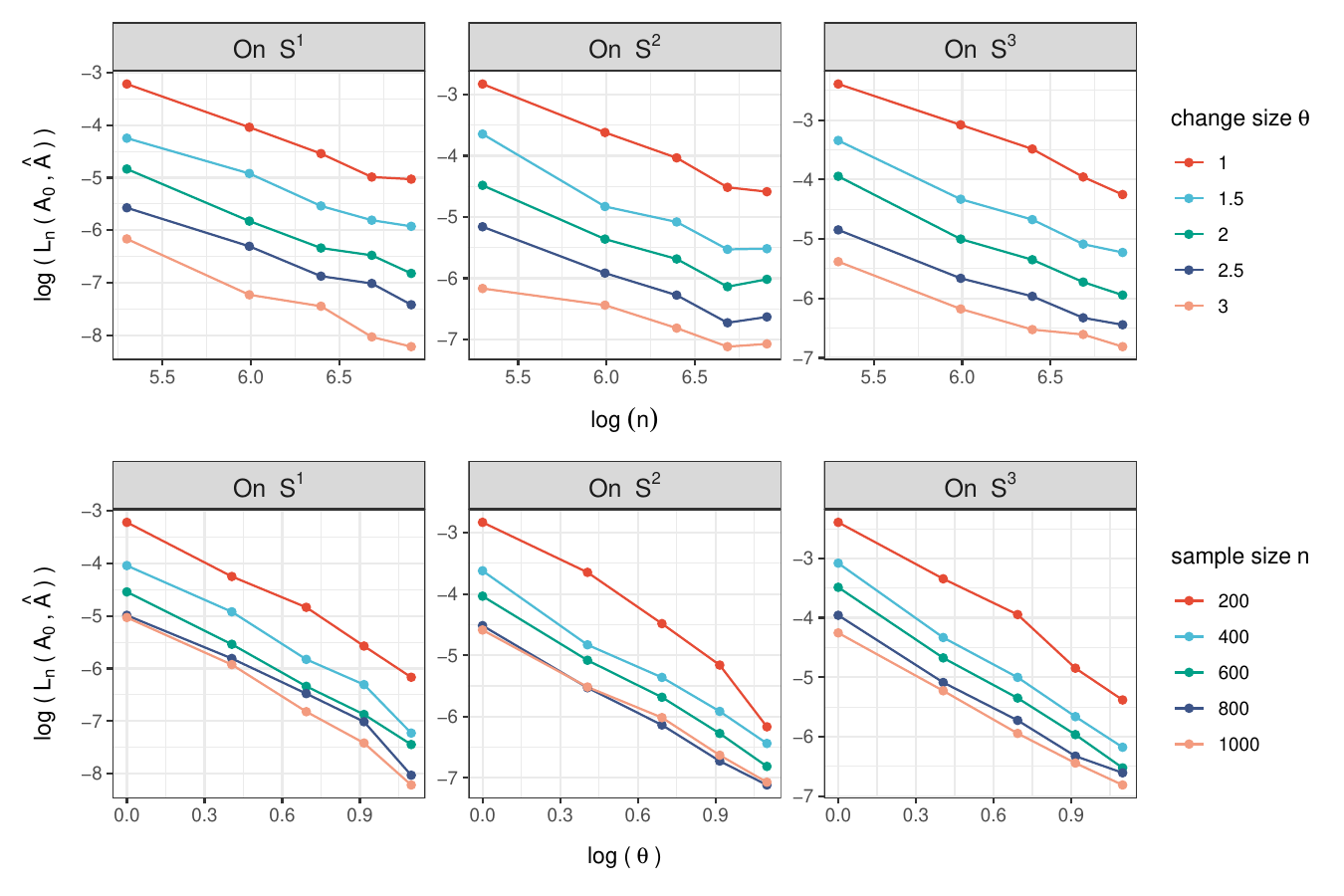}
  \caption{Empirical loss of single change region estimation using CRISP, averaged over 100 Monte Carlo repetitions, plotted against the sample size (top panels) and signal strengths (bottom panels) on log-log scale, for different dimensions. Data generating mechanism described in Section~\ref{Sec:EmpPerf}.}
  \label{fig:rate_singlecp}
\end{figure}
\begin{figure}[htbp]
  \centering
  \includegraphics[scale=0.6]{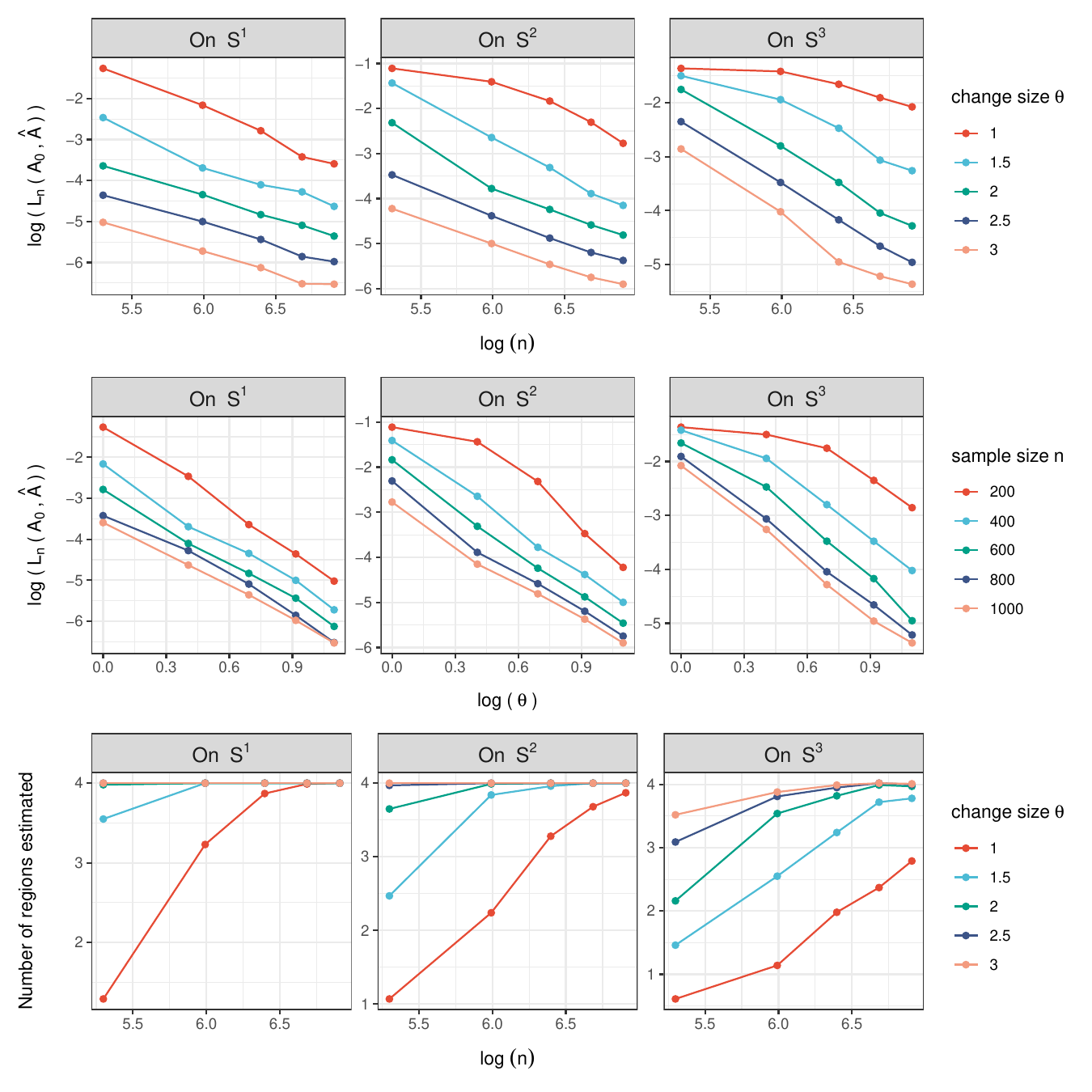}
  \caption{Empirical results for multiple change-region estimation using CRISP, averaged over 100 Monte Carlo repetitions. The loss is plotted against the sample size (top panels) and signal strengths (middle panels) on log-log scale, and the average number of estimated regions (true value: four) is shown in the bottom panels, for different dimensions. The data-generating mechanism is described in Section~\ref{Sec:EmpPerf}.} 
  \label{fig:rate_multicp}
  \vspace{2cm}
\end{figure}

We observe from both figures that as the dimension increases, the estimation problem becomes harder and larger sample size is needed to achieve similar statistical performance. As stated in the assumptions of Theorem~\ref{thm:multi_consistency}, the inequality $\omega^d \theta^2 \gtrsim \sqrt{1/n}$ highlights a fundamental trade-off between signal strength, region separation, and sample size. In a $d$-dimensional ambient space, the volume separating change regions scales like $\omega^{d-1}$, which decays exponentially with increasing $d$. As a result, even if the separation $\omega$ and signal strength $\theta$ are fixed, the left-hand side of the inequality diminishes rapidly as dimension grows. To maintain detectability,  the sample size $n$ must grow exponentially in $d$.

\subsection{Comparison with existing methods}

We compare our method to three point cloud segmentation methods~\citep{li2012new,silva2016imputation,dalponte2016tree}, all available from the $\texttt{R}$ package $\texttt{lidR}$. The comparison is made under the same settings of multiple regions detection described in Section~\ref{Sec:EmpPerf} for $d=3$ and $r=4$, since the competitors are designed for three-dimensional data. The average losses and the adjusted Rand index \citep{rand1971objective} are plotted in Figure~\ref{fig:rate_competitor}. We see that CRISP performs better than the competing methods under all settings. Part of the reason is that the competitors are designed to detect small regions (of flexible shape) of local elevations in the observations, which tend to overestimate the number of change regions due to its flexible nature.

\begin{figure}[htbp]
  \centering
  \includegraphics[scale=0.56]{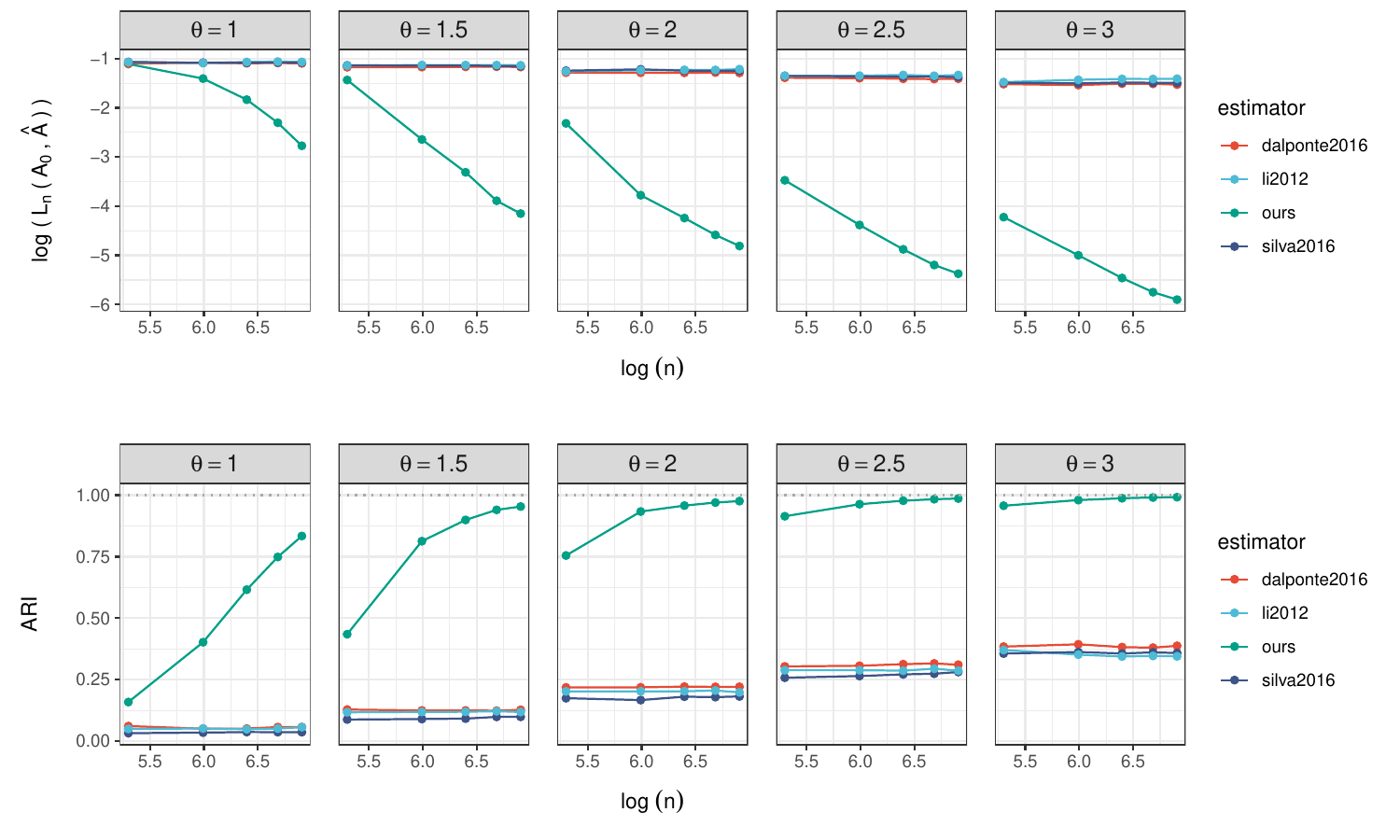}
  \caption{Empirical results for multiple change region estimation using CRISP and competitors from \texttt{lidR} averaged over $100$ Monte Carlo repetitions when $d=3$. Upper panel: losses plotted against the sample size $n$ on log-log scale. Lower panel: Adjusted Rand Index plotted against $n$ on log-log scale.}
  \label{fig:rate_competitor}
\end{figure}
\section{Real data applications}\label{sec:real_data}
\subsection{Global temperature change}
We study global near-surface temperature changes by combining land temperature dataset derived from ERA5 in \citet{cucchi2020wfde5} with sea-surface temperature dataset from the ESA CCI project~\citep{embury2024satellite}. For each of the years 1989, 1999, 2009 and 2019, we consider the December monthly mean on a regular $2^\circ\times 2^\circ$ latitude--longitude grid with about $16{,}200$ cells per year. Our objective is to detect spatially localised regions where the mean temperature change differs from the background when comparing consecutive decades.

For each decade pair, we form a global difference map by combining the land and sea datasets, replacing missing land values by the sea-surface values at the same grid cell. We then sample $n=2000$ uniform points on $\mathbb{S}^2$ and map each point to the nearest grid cell in longitude and latitude. This yields paired observations $(X_i,Y_i)$, where $X_i\in\mathbb{S}^2$ is the spatial location and $Y_i$ is the associated temperature difference. We standardise $Y_i$ by a robust Median Absolute Deviation (MAD) \citep{hampel1974influence} scale estimate $\hat\sigma$ based on neighbours induced by the minimum spanning tree (MST) on the sampled locations. Specifically, we first compute a typical neighbour distance $d_{\\max}$ as the mean edge length of the MST, and regard two observations as neighbours if their MST edge length is at most $d_{\\max}$. For each $i$, let $\mathcal{N}(i)$ denote the index set of such valid MST-neighbours of $X_i$. We then locally average each response with its valid neighbours by setting $\bar{Y}_i=\mathrm{mean}\bigl(\{Y_i\}\cup\{Y_j:j\in\mathcal{N}(i)\}\bigr)$ and apply the standard degrees-of-freedom correction factor $\sqrt{m/(m-1)}$ with $m=|\mathcal{N}(i)|+1$ to account for estimating the local mean from $m$ observations. We then use $Y_i/\hat\sigma$ as input to the detection procedure. We apply the CRISP procedure with post-processing to the three decadal difference maps (1989--1999, 1999--2009, and 2009--2019). Figure~\ref{Fig:realdata_result_temperature} shows the raw difference maps with the detected discs (blue discs) overlaid on the sampled locations. Comparing across decades, the detected regions are concentrated over Antarctica and the adjacent Southern Ocean in 1989--1999; in 1999--2009, the pattern shifts, with additional discs appearing over the Arctic; and in 2009--2019, the detected regions remain prominent over Central Asia.

\begin{figure}[!htbp]
  \centering
  \includegraphics[width=0.8\linewidth]{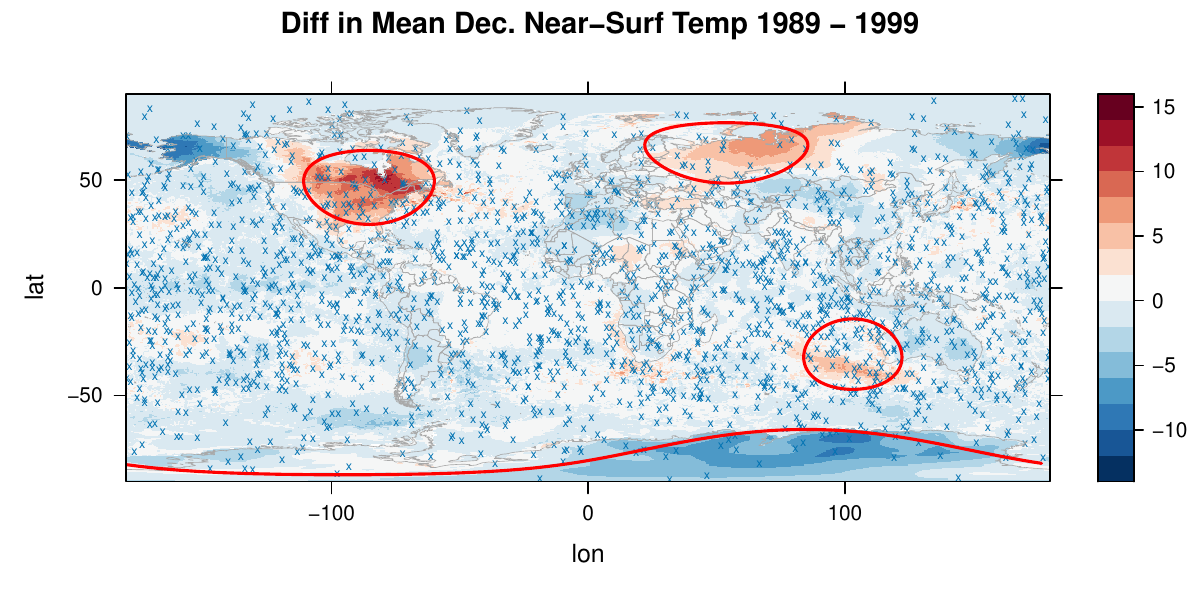}
  \includegraphics[width=0.8\linewidth]{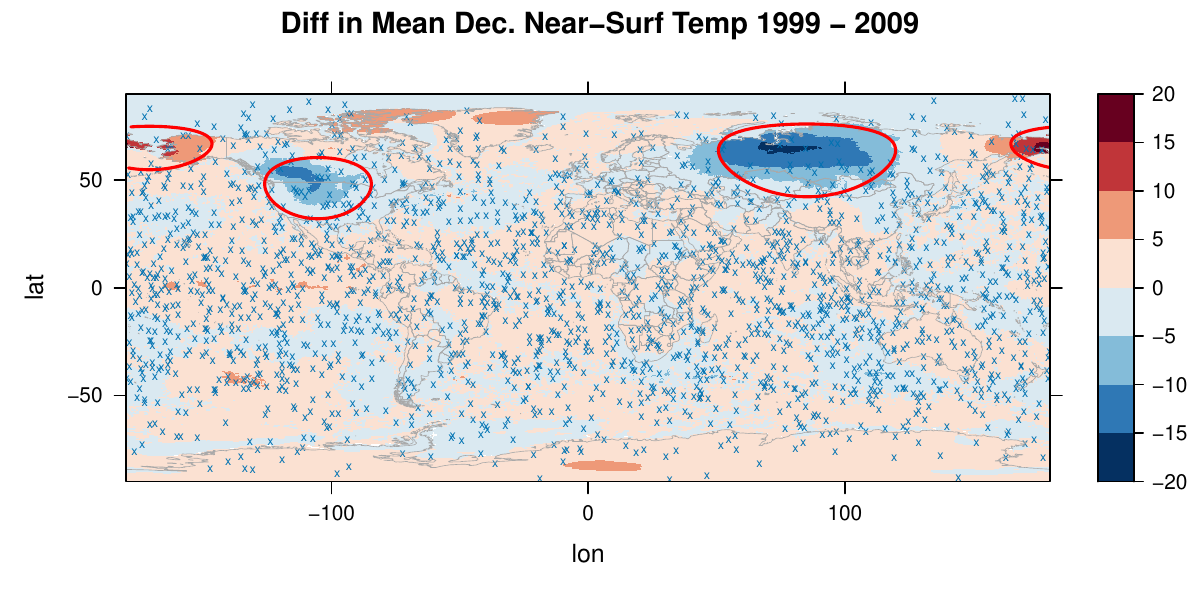}
  \includegraphics[width=0.8\linewidth]{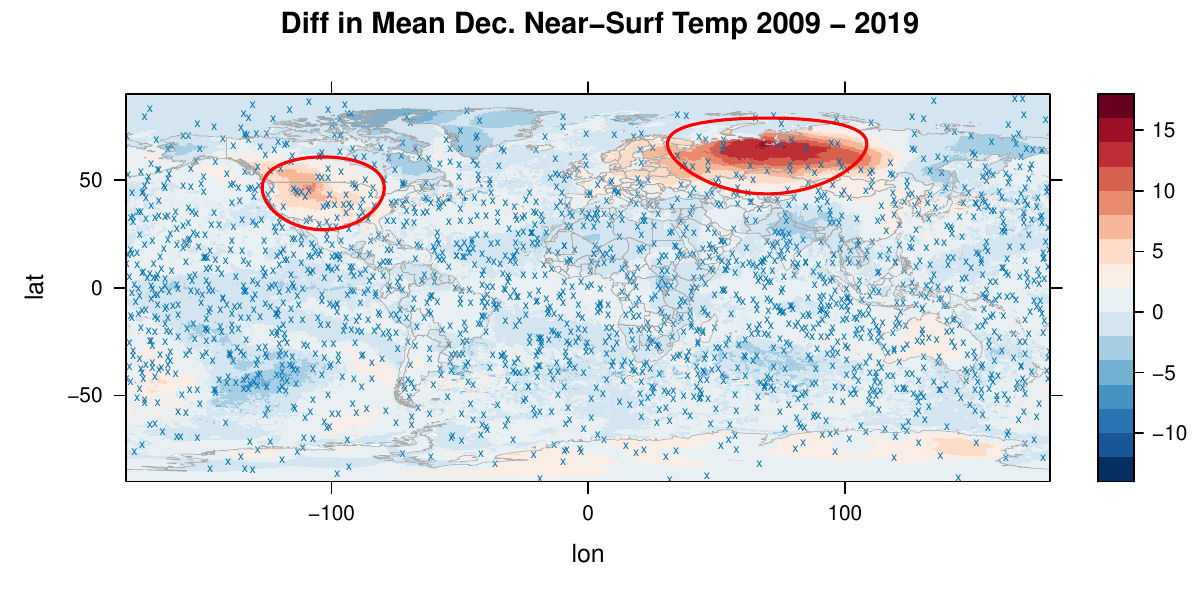}
  \caption{Maps of temperature difference between consecutive decades using the ERA5 and ESA data. Observations at 2000 randomly sampled grid points (blue crosses) are used to compute abnormal regions of temperature difference using the CRISP methodology. The estimated abnormal regions are shown as red discs on the maps.}
  \label{Fig:realdata_result_temperature}
\end{figure}

\subsection{Ozone Hole Monitoring}
We analyse total column ozone measurements from NASA's Aura Ozone Monitoring Instrument (OMI). The data are provided on a regular latitude--longitude grid and are reported in Dobson units (DU). We consider three calendar years (2005, 2015, and 2025) and focus on austral spring (September--November), averaging daily observations to form seasonal mean ozone fields. Our analysis is restricted to Southern Hemisphere high latitudes (60--90$^\circ$S), where Antarctic ozone depletion is known to occur; outside this band the field is comparatively homogeneous and would dilute the local signal of interest. 

We follow the same preprocessing steps as in the temperature analysis. In particular, we sample $1000$ locations uniformly on the sphere over latitudes below $ 60^\circ$S, snap each sampled location to the nearest latitude--longitude grid cell, and extract the corresponding ozone values. The extracted values are then rescaled using the same robust scale estimate as in the temperature example.

Applying our change region detection procedure identifies a coherent low-ozone region centred over Antarctica, with boundaries that closely track the spatial extent of the ozone hole. Figure~\ref{fig:ozone} overlays the detected regions on the corresponding ozone maps for 2005, 2015, and 2025. Across the three snapshots, the detected region remains localised over the polar cap, while its spatial extent varies from year to year. The bar-shaped artefacts in the background field are inherited from the original gridded visualisation rather than produced by our procedure, and likely reflect the underlying spatial allocation of measurements used to construct the dataset.
\begin{figure}[htbp]
  \centering
  \includegraphics[width=\linewidth]{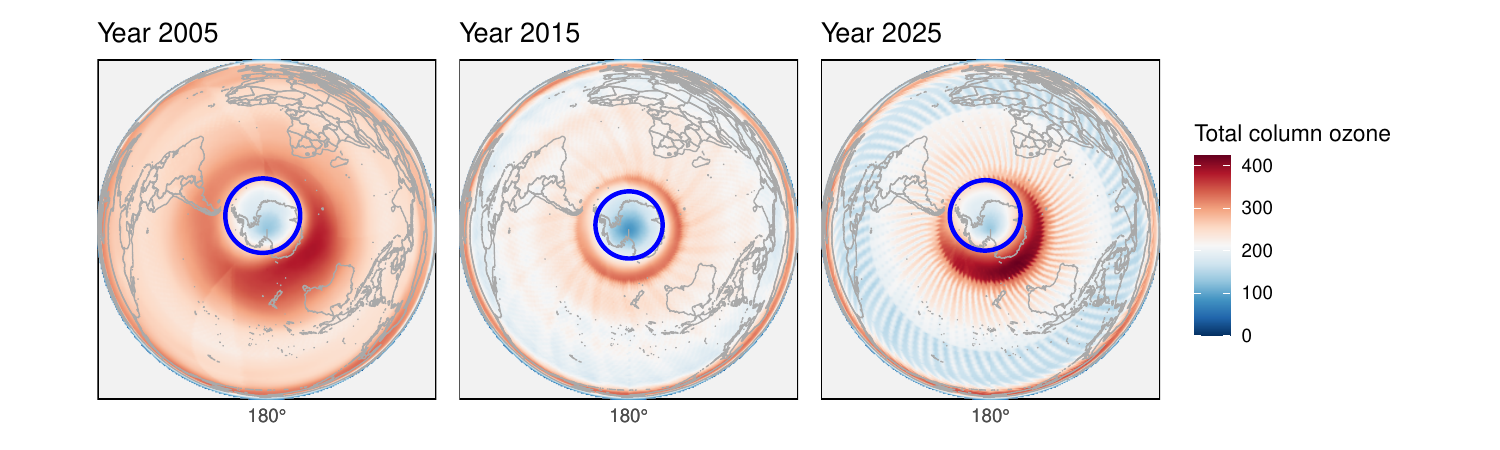}
  \caption{Ozone levels in the southern hemisphere in year 2005, 2015 and 2025. Detected abnormal regions using the CRISP method are shown as blue discs.}
  \label{fig:ozone}
\end{figure}
\vspace{0.5cm}

\section*{Acknowledgment}
The first author's research is funded by an LSE PhD Studentship. The authors have used large language models to refine the language of this manuscript, and have no conflicts of interest to declare.

\section*{Data Availability}
The data that support the findings of this study are available from the following resources available in the public domain: 

\begin{itemize}
\item Global temperature: \href{https://www.ecmwf.int/en/forecasts/dataset/ecmwf-reanalysis-v5}{https://www.ecmwf.int/en/forecasts/dataset/ecmwf-reanalysis-v5}  and \href{https://climate.esa.int/en/data/}{https://climate.esa.int/en/data/}
\item Ozone hole: \href{https://ozonewatch.gsfc.nasa.gov/data/omi/}{https://ozonewatch.gsfc.nasa.gov/data/omi/}
\end{itemize}

\newpage
\begin{appendix}
\section{Proofs of the main results}\label{appA}

In this section, we provide the proofs for Theorem~\ref{thm:consistency}, Corollary~\ref{Cor:EstimationDiscs} and Theorem~\ref{thm:multi_consistency}.

Throughout the proofs, without loss of generality, we assume the noise variance $\sigma^2=1$ is known in order to simplify the notation. With a different (or even unknown) $\sigma^2$, only minor modifications are needed for one to get through the proofs. Throughout the proof, we use $C_d$ to denote a constant whose value may change at different steps but is always independent of $n$ and depends only on $d$.
\subsection{Proof of Theorem~\ref{thm:consistency}}
The proof of Theorem~\ref{thm:consistency} will proceed in several steps. We first control the standardised sum of the noise uniformly over sets in $\mathcal{A}$. Recall the definition of $\mathrm{VCD}(\mathcal{A})$, the VC dimension of a family of sets $\mathcal{A}$ in the Introduction. 
\begin{proposition}
\label{prop:bound_sum_of_e_generic_set}
Let $\varepsilon_i\iid N(0,1)$, for $i\in[n]$ and let $\mathcal{D} = \{X_1,\ldots,X_n\}$ is a set of fixed design points on a manifold $\mathcal{M}$. Suppose that $\mathcal{A}$ is a family of subsets on $\mathcal{M}$. There exists a universal constant $C>0$ such that  
    \begin{align*}
    \mathbb{P}\biggl(\max_{A\in\mathcal{A}}\biggl| \frac{1}{\sqrt{|A|_{\mathcal{D}}}} \sum_{i:X_i \in A}\varepsilon_i\biggr|\geq\lambda\biggr)\leq Cn^{\text{VCD}(\mathcal{A})}e^{-\lambda^2/2},
    \end{align*}
(with the convention that $|A|_{\mathcal{D}}^{-1/2}\sum_{i:X_i\in A} \varepsilon_i = 0$ if $|A|_{\mathcal{D}} = 0$) and 
\[
\mathbb{P}\Bigl(\max_{A\in\mathcal{A}}\bigl| \mathcal{T}_A(\varepsilon)\bigr|\geq\lambda\Bigr)\leq Cn^{\text{VCD}(\mathcal{A})}e^{-\lambda^2/2}.
\]
\end{proposition}

\begin{proof}
Define $\mathcal{A}_{\mathcal{D}}:=\{A\cap\mathcal{D}: A\in\mathcal{A} \}$. By the Sauer--Shelah lemma \citep{sauer1972density}, we have 
$|\mathcal{A}_{\mathcal{D}}|\leq \sum_{j=0}^{\mathrm{VCD}(\mathcal{A})}\binom{n}{j} $, and
\begin{align*}
\mathbb{P}\biggl(\max_{A\in\mathcal{A}}\biggl\lvert\sum_{i:X_i\in A}\frac{\varepsilon_i}{\sqrt{|A|_{\mathcal{D}}}}\biggr\rvert\geq\lambda\biggr) &= \mathbb{P}\biggl(\max_{A\in\mathcal{A}_{\mathcal{D}}}\biggl\lvert\sum_{i:X_i\in A}\frac{\varepsilon_i}{\sqrt{|A|_{\mathcal{D}}}}\biggr\lvert\geq\lambda\biggr)\\
        &\leq |\mathcal{A}_{\mathcal{D}}|\,\mathbb{P}(|Z|\geq\lambda) \leq|\mathcal{A}_{\mathcal{D}}|\,e^{-\lambda^2/2} \lesssim n^{\mathrm{VCD}(\mathcal{A})} e^{-\lambda^2/2},
    \end{align*}
as desired, where the penultimate inequality follows from the standard Gaussian tail bound. Similarly, since $\mathcal{T}_A(\varepsilon) \sim N(0,1)$ for any $A\in\mathcal{A}$, we have 
\begin{align*}
    \mathbb{P}\Bigl(\max_{A\in\mathcal{A}}\bigl| \mathcal{T}_A(\varepsilon)\bigr|\geq\lambda\Bigr) &= \mathbb{P}\Bigl(\max_{A\in\mathcal{A}_\mathcal{D}}\bigl| \mathcal{T}_A(\varepsilon)\bigr|\geq\lambda\Bigr) \leq |\mathcal{A}_{\mathcal{D}}| \;  \mathbb{P}(|Z|\geq \lambda) \\
    &\lesssim n^{\mathrm{VCD}(\mathcal{A})}e^{-\lambda^2/2},
\end{align*}
which completes the proof.
\end{proof}

Proposition~\ref{prop:bound_sum_of_e_generic_set} allows us to control the CUSUM statistics of the noise $\varepsilon = (\varepsilon_1,\ldots,\varepsilon_n)$ on various candidate change regions. Proposition \ref{prop:cusum_diff_e<L} below further establishes that such CUSUM statistics are close to each other if the candidate change regions are close in terms of the error $L_n$ defined in~\eqref{eqn:def:loss}.
\begin{proposition}
\label{prop:cusum_diff_e<L}
    Let $\mathcal{D}=\{X_1,\ldots,X_n\}$ be deterministic design points on a manifold $\mathcal{M}$. Let $A_1,A_2\in\mathcal{A}$ be such that $L_n \{A_1,\emptyset\}\geq \tau>0$. Fix $\lambda > 0$, for any $A\in\mathcal{A}$, we denote
    \[
    \mathcal{E}_{n,A} := \biggl\{\biggl|\sum_{i:X_i\in A} \varepsilon_i\biggr| \leq \lambda |A|_{\mathcal{D}}^{1/2}\biggr\},
    \]
    and write $\mathcal{E}:= \mathcal{E}_{n, A_1}\cap \mathcal{E}_{n, A_2} \cap \mathcal{E}_{n, A_1\cap A_2} \cap \mathcal{E}_{n, A_1\setminus A_2}\cap \mathcal{E}_{n, A_2\setminus A_1}\cap \mathcal{E}_{n,\mathcal{M}}$. Furthermore, denote $\mathcal{F}_{n,A}:=\{|\mathcal{T}_A(\varepsilon)|\leq \lambda\}$ and write $\mathcal{F} := \mathcal{F}_{n,A_1}\cap \mathcal{F}_{n,A_2}$.
    
    Then, on the event $\mathcal{E} \cap \mathcal{F}$, we have 
    \begin{align*}
        |\mathcal{T}_{A_1}(\varepsilon) - \mathcal{T}_{A_2}(\varepsilon)|&\leq \min\biggl\{\frac{3\lambda}{\tau} L_n^{1/2}(A_1,A_2),\, 2\lambda\biggr\} .
    \end{align*}
\end{proposition}

\begin{proof}
We first assume that $L_n(A_1, A_2)\leq \tau/2$. The CUSUM statistic for any candidate change region $A$ may be re-written as 
\begin{align*}
	\mathcal{T}_A(\varepsilon) &= \sqrt\frac{|A|_{\mathcal{D}}|A^{\mathrm{c}}|_\mathcal{D}}{n} \biggl(\frac{1}{|A|_\mathcal{D}}\sum_{i:X_i \in A} \varepsilon_i - \frac{1}{|A^c|_{\mathcal{D}}}\sum_{i: X_i \in A^{\mathrm{c}}}\varepsilon_i\biggr)\\
	& =\sqrt\frac{|A|_{\mathcal{D}} n}{|A^{\mathrm{c}}|_{\mathcal{D}}} \biggl(\frac{1}{|A|_{\mathcal{D}}}\sum_{i:X_i \in A}\varepsilon_i-\frac{1}{n}\sum_{i=1}^n \varepsilon_i\biggr).
\end{align*}

Writing $m_1 := |A_1|_{\mathcal{D}}$ and $m_2 =: |A_2|_{\mathcal{D}}$. Without loss of generality, assume that $m_1 \leq m_2$. The above displayed equation allows us to decompose
\[  
\mathcal{T}_{A_1}(\varepsilon) - \mathcal{T}_{A_2}(\varepsilon) = I_1 + I_2 + I_3,
\]
where 
\begin{align*}
    I_1 & := \biggl(\sqrt\frac{m_2}{n-m_2} - \sqrt\frac{m_1}{n-m_1}\biggr) \sqrt\frac{1}{n}\sum_{i=1}^n \varepsilon_i,\\
    I_2 & := \biggl(\sqrt\frac{n}{m_1(n-m_1)} - \sqrt\frac{n}{m_2(n-m_2)} \biggr)\sum_{i: X_i \in A_1\cap A_2} \varepsilon_i,\\
    I_3 & := \sqrt\frac{n}{m_1(n-m_1)}\sum_{i:X_i\in A_1\setminus A_2} \varepsilon_i - \sqrt\frac{n}{m_2(n-m_2)} \sum_{i:X_i\in A_2\setminus A_1}\varepsilon_i.
\end{align*}
We bound the three terms separately. To control $I_1$, let $f(x):=\sqrt{x /(n-x)}$, then by the Mean Value Theorem, on the event $\mathcal{E}_{n, \mathcal{M}}$, we have  
\begin{align}
\label{proof:prop:cusum_noise:eqn_I1}
|I_1| &\leq \lambda (m_2-m_1) \sup_{u\in[m_1,m_2]} |f'(u)| = \frac{(m_2-m_1)\lambda}{2} \sup_{u\in[m_1,m_2]} \sqrt\frac{n^2}{u(n-u)^3} \nonumber\\
&\leq \frac{\lambda(m_2-m_1)}{2^{1/2}\tau^{3/2}n} \leq  \frac{\lambda}{2\tau} L_n^{1/2}(A_1,A_2),
\end{align}
where we used the fact $L_n(A_1,\emptyset) = \min\{m_1/n, 1-m_1/n\} \geq \tau$ and the assumption $L_n(A_1,A_2)\leq \tau/2$ in the penultimate inequality.

To control $I_2$, let $g(x):=\sqrt{n/\{x(n-x)\}}$. Again by a similar argument as above, on the event $\mathcal{E}_{n,A_1\cap A_2}$, we have
\begin{align}
\label{proof:prop:cusum_noise:eqn_I2}
    |I_2| &\leq \lambda \sqrt{|A_1\cap A_2|_{\mathcal{D}}} (m_2-m_1)\sup_{u\in[m_1,m_2]}|g'(u)|\nonumber\\
    &\leq\frac{\lambda n^{1/2}(m_2-m_1)}{2}\sup_{u\in[m_1,m_2]}\frac{|n-2u|}{2n}\biggl(\frac{n}{u(n-u)}\biggr)^{3/2} \leq \frac{\lambda}{2\tau} L_n^{1/2}(A_1,A_2).
\end{align}

To control $I_3$, on the event $\mathcal{E}_{n,A_1\setminus A_2}\cap \mathcal{E}_{n,A_2\setminus A_1}$, it holds that
\begin{align}\label{proof:prop:cusum_noise:eqn_I3}
	|I_3|& \leq \sqrt\frac{n}{m_1(n-m_1)}|A_1\setminus A_2|_{\mathcal{D}}^{1/2}\lambda + \sqrt\frac{n}{m_2(n-m_2)}|A_2\setminus A_1|_{\mathcal{D}}^{1/2}\lambda \nonumber \\
    &\leq 
    \sqrt{\frac{2}{n \tau}}\Bigl(\sqrt{|A_1\setminus A_2|_{\mathcal{D}}}+\sqrt{|A_2\setminus A_1|_{\mathcal{D}}}\Bigr) \lambda   \leq \frac{2}{\sqrt{ \tau}} L_n^{1/2}(A_1,A_2)  \lambda,
\end{align}
where we used the fact that $\sqrt{a}+\sqrt{b} \leq \sqrt{2(a+b)}$ for $a, b \geqslant 0$ in the last inequality.

Combining~\eqref{proof:prop:cusum_noise:eqn_I1},~\eqref{proof:prop:cusum_noise:eqn_I2} and~\eqref{proof:prop:cusum_noise:eqn_I3}, we get that 
\[
\bigl|\mathcal{T}_{A_1}(\varepsilon)-\mathcal{T}_{A_2}(\varepsilon)\bigr| \leq\frac{3\lambda}{\tau}L_n^{1/2}(A_1,A_2),
\]
as desired.

It remains to handle the case where $L_n(A_1,A_2) > \tau/2$. In this case, we have on the event $\mathcal{F}$ that 
\begin{align*}
    \bigl|\mathcal{T}_{A_1}(\varepsilon)-\mathcal{T}_{A_2}(\varepsilon)\bigr| &\leq |\mathcal{T}_{A_1}(\varepsilon)| + \bigl|\mathcal{T}_{A_2}(\varepsilon)| \leq 2\lambda,
\end{align*}
which completes the proof.
\end{proof}
In complement to Proposition~\ref{prop:cusum_diff_e<L}, Proposition~\ref{Prop:CUSUM_Signal} below states that, the signal CUSUM statistic evaluated at a candidate change region $A$,  decreases from its peak at $A=R_1$ at a linear rate of $L_n(R_1,A)$.
\begin{proposition}
\label{Prop:CUSUM_Signal}
Let $X_1,\ldots,X_n$ be deterministic design points on a manifold $\mathcal{M}$. For a change region $R_1\subseteq \mathcal{M}$, we have a vector $\mu = (\mu_1,\ldots,\mu_n)^\top$ such that $\mu_i = \mu^{(1)}$ if $X_i\in R_1$ and $\mu_i = \mu^{(2)}$ if $X_i\notin R_1$. Let $\theta := |\mu^{(1)}-\mu^{(2)}|$. If $L_n(R_1,\emptyset) = \tau>0$, then for any subset $A\subseteq \mathcal{M}$, we have 
\begin{align*}
    \mathcal{T}_{R_1}(\mu)-\mathcal{T}_{A}(\mu) \geq \frac{\sqrt{n}\theta}{4}\min\biggl\{ \frac{L_n(A, R_1)}{\sqrt{\tau}}, \sqrt{\tau}\biggr\}.
\end{align*}
\end{proposition}
\begin{proof}
Since both $L_n(R_1, A)$ and $\mathcal{T}_{R_1}(\mu)$ are invariant to replacing $R_1$ by $R_1^{\mathrm{c}}$, we may assume without loss of generality that $|R_1\triangle A|_{\mathcal{D}}\leq n/2$.

For any vector $z = (z_1,\ldots,z_n)^\top$ and any $A\in\mathcal{A}$, we define
\[
\mathrm{RSS}_{A}(z) := \sum_{i:X_i\in A} (z_i - \bar z_{A})^2 + \sum_{i: X_i\in A^{\mathrm{c}}} (z_i - \bar z_{A^{\mathrm{c}}})^2,
\]
where $\bar z_A:= |A|_{\mathcal{D}}^{-1}\sum_{i: X_i\in A}z_i$ and $\bar z_{A^{\mathrm{c}}}:= |A^{\mathrm{c}}|_{\mathcal{D}}^{-1}\sum_{i: X_i\in A^{\mathrm{c}}}z_i$ (with the convention that $\bar z_A = 0$ if $|A|_{\mathcal{D}} = 0$ and $\bar z_{A^{\mathrm{c}}} = 0$ if $|A^{\mathrm{c}}|_{\mathcal{D}} = 0$). We also write $\mathrm{RSS}(z) := \mathrm{RSS}_{\emptyset}(z)$.

We observe that 
\[
\mathcal{T}_{A}^2(\mu) = \mathrm{RSS}(\mu)-\mathrm{RSS}_{A}(\mu).
\]
Hence the difference of the squared CUSUM is equal to the difference of the residual sum of squares; for detailed discussions, see Lemma 4 of \cite{baranowski2019narrowest}.  Notice also that $\mathrm{RSS}_{R_1}(\mu)=0$. Together with Lemma \ref{Lemma:CUSUM_signal}(b), for any $A\subseteq \mathcal{M}$ such that $|A\triangle R_1|_{\mathcal{D}}\leq n/2$, we have
	\begin{align*}
		\mathcal{T}_{R_1}^2(\mu) -\mathcal{T}_{A}^2(\mu)&=\mathrm{RSS}_{A}(\mu) - \mathrm{RSS}_{R_1}(\mu)=\mathrm{RSS}_{A}(\mu)
        \\
        &\geq \frac{\theta^2}{2}\min\{|A\triangle R_1|_{\mathcal{D}},|(A\triangle R_1)^{\mathrm{c}}|_{\mathcal{D}}, |R_1|_{\mathcal{D}},|R_1^{\mathrm{c}}|_{\mathcal{D}}\}\\
        &=\frac{\theta^2}{2}\min\{|A\triangle R_1|_{\mathcal{D}},n\tau\},
	\end{align*}
	where the last equality is a result of the assumption that $|R_1\triangle A|_{\mathcal{D}}\leq n/2$ and the definition of $\tau$. Again, by Lemma~\ref{Lemma:CUSUM_signal}(a,c), we have that $\mathcal{T}_{R_1}(\mu) + \mathcal{T}_A(\mu) \leq 2\mathcal{T}_{R_1}(\mu) \leq 2\theta (n\tau)^{1/2}$. Consequently, we have 
    \[  
    \mathcal{T}_{R_1}(\mu) -\mathcal{T}_{A}(\mu) \geq \frac{\theta}{4}\min\biggl\{\frac{|A\triangle R_1|_{\mathcal{D}}}{\sqrt{n\tau}}, \sqrt{n\tau}\biggr\}
    \]
    as desired.
\end{proof}
\begin{proof}[Proof of Theorem~\ref{thm:consistency}]
By the definition of $\hat{R}$ we have $\mathcal{T}_{R_1}(Y) \leq \mathcal{T}_{\hat{R}}(Y)$. Hence
	\begin{align*}
		\mathcal{T}_{R_1}(\mu)-\mathcal{T}_{\hat{R}}(\mu)&=\mathcal{T}_{R_1}(Y)-\mathcal{T}_{\hat{R}}(Y)+\mathcal{T}_{\hat{R}}(\varepsilon)-\mathcal{T}_{R_1}(\varepsilon)\leq \mathcal{T}_{\hat{R}}(\varepsilon)-\mathcal{T}_{R_1}(\varepsilon).
	\end{align*}
Let $\mathcal{E}_{n,A}$ and $\mathcal{F}_{n,A}$, which depends on some $\lambda$ that will be specified later, be as defined in Proposition~\ref{prop:cusum_diff_e<L}.
By Propositions~\ref{prop:cusum_diff_e<L} and~\ref{Prop:CUSUM_Signal}, on the event 
\[
\Omega := \bigcap_{A\in\mathcal{A} \cup \{\mathcal{M}\}} \biggl(\mathcal{E}_{n,A} \cap \mathcal{E}_{n,A\cap R_1} \cap \mathcal{E}_{n,A\setminus R_1}\cap \mathcal{E}_{n,R_1\setminus A}\biggr) \cap \bigcap_{A\in\mathcal{A}}\mathcal{F}_{n,A} 
\]
we have 
\begin{equation}
\frac{\sqrt{n}\theta}{4}\min\biggl\{ \frac{L_n(\hat R, R_1)}{\sqrt{\tau}}, \sqrt{\tau}\biggr\} \leq |\mathcal{T}_{\hat{R}}(\varepsilon)-\mathcal{T}_{R_1}(\varepsilon)| \leq \min\biggl\{\frac{3\lambda}{\tau}L_n^{1/2}(\hat R, R_1), \,2\lambda\biggr\}.\label{Eq:Sandwich}
\end{equation}
We choose $\lambda = \sqrt{4\sigma^2\mathrm{VCD}(\mathcal{A})\log n}$. First note that if $L_n(\hat R, R_1) > \tau/2$, then we have $\sqrt{n\tau}\theta/8 \leq 2\lambda$, which contradicts the assumption in the theorem if we choose $C_0 = 1024$. Hence we may assume that  $L_n(\hat R, R_1) \leq \tau/2$, then~\eqref{Eq:Sandwich} implies that 
\[
L_n(\hat R, R_1) \leq \frac{144\lambda^2}{n\tau\theta^2},
\]
as desired.

It remains to control the probability of $\Omega$. To this end, define $\mathcal{A}\cap_{*} R_1:=\{A\cap R_1:A\in\mathcal{A}\}$, $\mathcal{A}\setminus_{*} R_1:=\{A\setminus R_1:A\in\mathcal{A}\}$, $R_1\setminus_{*} \mathcal{A}:=\{R_1\setminus A:A\in\mathcal{A}\}$. We apply Proposition~\ref{prop:bound_sum_of_e_generic_set} with a union bound to get for some universal constant $C_1>0$ that 
\[
\mathbb{P}(\Omega^{\mathrm{c}}) \leq Cn^D e^{-\lambda^2/2},
\]
where 
\[
D := \max\{\mathrm{VCD}(\mathcal{A}), \mathrm{VCD}(\mathcal{A}\cap_* R_1),\mathrm{VCD}(\mathcal{A}\setminus_{*} R_1), \mathrm{VCD}(R_1\setminus_{*} \mathcal{A})\}.
\]
The proof is complete since $D = \mathrm{VCD}(\mathcal{A})$ 
by Lemma~\ref{Lemma:vcdOperationDiffCap}.
\end{proof}

\subsection{Proof of Corollary \ref{Cor:EstimationDiscs}}
\begin{proof}
By Lemma~\ref{lemma:vcdim_discs}, $\mathrm{VCD}(\mathcal{A})=d+1$ and the rate of convergence of the loss follows. We proceed to show the rate of parameter estimation. By~\citet[Theorem 2]{vapnik2015uniform}, with probability at least $1-n^{-2(d+1)}$, we have
    \[{|\hat{R}\triangle R_1|_D} \geq\frac{n\mu(\hat{R}\triangle R_1)}{\mu(\mathbb{S}^{d-1})}-14\sqrt{nd\log(n)},\]
    where $\mu$ denotes the Lebesgue measure on the sphere. 
 Together with (a), there are universal constants $C_1,C_2$ such that
\begin{align*}
	& \mathbb{P}\left(\frac{\mu(\hat{R} \triangle R_1)}{\mu(\mathbb{S}^{d-1})} \leqslant \frac{C_2 \sigma^2 d \log (n)}{\theta^2\tau n}+14\sqrt{d\log(n)/n}\right)  \geqslant 1-C_1 n^{-(d+1)}.
\end{align*} 
By Lemma \ref{lemma:param_est_lb}, we have
$$
\begin{gathered}
	 \mathbb{P}\left(\|\alpha-\alpha'\|+|\beta-\beta'|\leqslant C_{d,\beta}\Big(\frac{C_2 \sigma^2 d \log (n)}{\theta^2\tau n}+\sqrt{d\log(n)/n}\Big)\right)  \geqslant 1-C_1 n^{-(d+1)},
\end{gathered}
$$
for constants $C_1,C_2,C_{d,\beta}$.
 
\end{proof}
\subsection{Proof of Theorem \ref{thm:multi_consistency}}
We first state an analog of Proposition~\ref{prop:bound_sum_of_e_generic_set} for the local CUSUM statistic of noise $\mathcal{T}_{A}^B(\varepsilon)$ uniformly over all $A$ and $B$.
\begin{proposition}\label{Cor:distLocalCUSUM}
Suppose $\mathcal{M} = \mathbb{S}^{d-1}$ and $\mathcal{S}$ is the class of all discs on $\mathcal{M}$. Then for some universal constant $C>0$, we have
	\begin{equation*}
		\mathbb{P}\Bigl(\max_{A,B \in \mathcal{S}: A \subseteq B} \bigl|\mathcal{T}_A^{B}(\varepsilon)\bigr| \geq \lambda \Bigr)  \leq Cn^{2d+2} e^{-\lambda^2/2}.
	\end{equation*}
\end{proposition}
\begin{proof}
	Define $\mathcal{S}_{\mathcal{D}}:=\{A\cap\mathcal{D}: A\in\mathcal{S} \}$ and by Sauer-Shelah lemma we have $|\mathcal{S}_{\mathcal{D}}|\lesssim n^{\mathrm{VCD}(\mathcal{S})}$. By a union bound, we have
	\begin{align*}
		\mathbb{P}\Bigl(\max_{A,B \in \mathcal{S}: A \subseteq B}\bigl| \mathcal{T}^{B}_A(\varepsilon)\bigr|\geq\lambda\Bigr) \leq \sum_{B\in\mathcal{S}_{\mathcal{D}}} \mathbb{P}\Bigl(\max_{A \subseteq B}\bigl| \mathcal{T}^{B}_A(\varepsilon)\bigr|\geq\lambda\Bigr) \lesssim  n^{2\mathrm{VCD}(\mathcal{S})}e^{-\lambda^2/2},
	\end{align*}
    where the second inequality follows from Proposition~\ref{prop:bound_sum_of_e_generic_set} using $B$ in place of $\mathcal{M}$ therein. The desired conclusion follows since $\mathrm{VCD}(\mathcal{A}) = d+1$ when $\mathcal{S}$ is the set of discs.
\end{proof}
\begin{proof}[Proof of Theorem~\ref{thm:multi_consistency}]   
 Define events 
    \begin{align*}
            \mathcal{E}_{1} &:= \bigl\{ \forall k\in[r], \exists j \in[J_n] \text{ s.t. } R_k \subseteq B_j,\mathrm{dist}(R_k,B_j^{\mathrm{c}})\geq\omega/2, R_i \cap B_j = \emptyset , \forall i \neq k\bigr\},\\
        \mathcal{E}_{2}&:=\bigl\{\max_{A,B\in\mathcal{S}: A \subseteq B} \left|\mathcal{T}_A^{B}(\varepsilon)\right| \leq \sqrt\lambda \bigr\},\\
       	\mathcal{E}_3&:=\bigcap_{A,B\in\mathcal{S}:A\subseteq B} \Bigl\{\bigl|\mathrm{RSS}^{B}_{A}(Y)-(\mathrm{RSS}^{B}_{A}(\mu)+m_{A,B})\bigr|\\
        &\hspace{7cm}\leq2\sqrt{2\lambda(\mathrm{RSS}^{B}_{A}(\mu)+m_{A,B}/2)}+2\lambda\Bigr\},\\
        \mathcal{E}_4&:=\biggl\{\sup_{A\in\mathcal{S}} \biggl|\frac{|A|_{\mathcal{D}}}{n} - P_X(A)\biggr| \leq 7\sqrt\frac{\lambda}{n}\biggr\}.
    \end{align*}
    By Lemma~\ref{proof:lemma_exact_one_disc_covered}, we have $\mathbb{P}(\mathcal{E}_1) \geq 1-re^{-C_dJ_n\omega^d}$ for some constant $C_d$ depending only on $d$. By Proposition~\ref{Cor:distLocalCUSUM}, we have $\mathbb{P}(\mathcal{E}_2) \geq 1-Cn^{-2(d+1)}$ for some universal constant $C>0$. By \citet[Lemma~8.1]{Birge_2001} and the fact that $|\{A\cap \mathcal{D}:A\in\mathcal{S}\}|\lesssim n^{\mathrm{VCD}(\mathcal{S})} = n^{d+1}$, we have $\mathbb{P}(\mathcal{E}_3) \geq1-2n^{-2(d+1)}$. By~\citet[Theorem 2]{vapnik2015uniform} and the fact that $\mathrm{VCD}(\mathcal{S})=d+1$, we have $\mathbb{P}(\mathcal{E}_{4})\geq1-n^{-2(d+1)}$.
    Thus, we have
\[
\mathbb{P}(\mathcal{E}_1\cap\mathcal{E}_2\cap\mathcal{E}_3\cap\mathcal{E}_4)\geq 1 - re^{-C_dJ_n\omega^d} - Cn^{-2d-2},
\]
for some universal constant $C>0$ and some constant $C_d$ depending only on $d$. Henceforth, we work on the event $\mathcal{E}_1\cap\mathcal{E}_2\cap\mathcal{E}_3\cap\mathcal{E}_4$, and prove the theorem in three steps. In Steps 1 and 2, we prove that every element of $\hat{\mathcal{R}}_{\mathrm{init}}$ in Algorithm~\ref{algo:multiple_cp_detection} 
estimates one of the true change regions with small error. In Step 3, we prove that each of the $r$ change regions is estimated by at least one element of $\hat{\mathcal{R}}_{\mathrm{init}}$. We complete the proof in Step 4 by showing that $r=\hat{r}=:|\hat{\mathcal{R}}|$ holds.

\textbf{Step 1:} 
Fix $j\in [J_n]$ and suppose that $\hat R_j\in\hat{\mathcal{R}}_{\mathrm{init}}$. Denote $\tilde{R}_j :=\cup_{\ell\in[r]}(B_j\cap R_\ell)$, in this step, we show that $L_n(\hat{R}_j,\tilde R_j)$ is small. As in the description of Algorithm~\ref{algo:multiple_cp_detection}, we abbreviate $\mathcal{T}_j := \mathcal{T}^{B_j}_{\hat R_j}$ and $\mathrm{RSS}_j := \mathrm{RSS}^{B_j}_{\hat R_j}$. By the fact that $\hat R_j$ is selected into $\hat{\mathcal{R}}_{\mathrm{init}}$, we have $\mathcal{T}_j(Y) \geq \lambda_T = 4(n\lambda)^{1/4}+8\sqrt{\lambda}$ and $\mathrm{RSS}_j(Y) \leq \gamma_j = m_j + 2\sqrt{m_j\lambda} + 2\lambda$.

 On the event $\mathcal{E}_2$, we have by the condition on $\lambda_T$ that 
	 \begin{align}\label{Eqn:TjMuLarge}
	 	T_j(\mu) \geq T_j(Y) - \sqrt{\lambda} >4(n\lambda)^{1/4}+6\sqrt{\lambda}.
	 \end{align}
	 On the event $\mathcal{E}_3$, we have
	 \[\mathrm{RSS}_j(Y)\geq\Bigl(\sqrt{\mathrm{RSS}_j(\mu)+m_j/2}-\sqrt{2\lambda}\Bigr)^2+m_j/2-4\lambda,\]
	 which implies that
	 \begin{align}\label{Eqn:RSSjMuSmall}
	 \mathrm{RSS}_j(\mu)&\leq\mathrm{RSS}_j(Y)+6\lambda-m_j+2\sqrt{2\lambda(\mathrm{RSS}_j(Y)+4\lambda-m_j/2)}\nonumber\\
	 	&\leq 2\sqrt{m_j\lambda}+8\lambda+2\sqrt{2\lambda(3m_j/2+7\lambda)} \leq 6\sqrt{n\lambda} + 16\lambda,
	 \end{align} 
 where we used the assumption that $\mathrm{RSS}_j(Y)\leq m_j+2\sqrt{m_j\lambda}+2\lambda\leq 2m_j+3\lambda$ in the last two inequalities and the fact that $m_j\leq n$ in the final bound. Combining~\eqref{Eqn:TjMuLarge} and~\eqref{Eqn:RSSjMuSmall}, we get
 \begin{equation}\label{Eqn:RSSsq<T}
 	\sqrt{2\mathrm{RSS}_j(\mu)}<\mathcal{T}_j(\mu).
 \end{equation} 
 From Lemma~\ref{Lemma:CUSUM_signal}(d) and (e), we have
 \begin{align}
 	\frac{2\mathrm{RSS}_j(\mu)}{n\theta^2}
    &\geq \frac{1}{n}\min\biggl\{\frac{(\mathcal{T}_j(\mu))^2}{\theta^2},|B_j\setminus \tilde R_j|_{\mathcal{D}},|\hat{R}_j \triangle \tilde{R}_j|_{\mathcal{D}}, \,|\hat R_j\triangle (B_j\setminus \tilde{R}_j)|_{\mathcal{D}}\biggr\}.\label{Eqn:RSS>Ln}
 \end{align} 
  By~\eqref{Eqn:RSSsq<T}, the first term in the minimum is larger than the left-hand side. For the second and the fourth term,  Lemma~\ref{Lemma:ChompingDoughnut} implies that there exists a disc $A$ of radius at least $\omega/8$ such that $A\subset  (B_j\setminus \tilde{R}_j)\setminus \hat{R}_j \subseteq (B_j\setminus \tilde R_j)\cap (\hat R_j\triangle (B_j\setminus \tilde{R}_j))$. On the event $\mathcal{E}_4$,
	\begin{align}\label{Eqn:SymDifLarge}
		\min\{|B_j\setminus \tilde R_j|_{\mathcal{D}}, |\hat R_j\triangle(B_j\setminus\tilde{R}_j)|_{\mathcal{D}}\} &\geq	  n P_X(A) - 7\sqrt{n\lambda} \nonumber\\
        &\geq n\underline{f}_XC_d(\omega/8)^d-7\sqrt{n\lambda}\nonumber\\
		&\stackrel{\mathrm{(i)}}\geq \frac{n}{2}\underline{f}_XC_d(\omega/8)^d\stackrel{\mathrm{(ii)}}\geq \frac{50\sqrt{n\lambda}}{\theta^2} >\frac{2\mathrm{RSS}_j(\mu)}{\theta^2},
	\end{align}
	where (i) follows from the assumption that $\omega\geq C_1(\lambda/n)^{1/(2d)}$ for a sufficiently large choice of $C_1>0$, (ii) follows from the assumption that $\omega^{d}\min_{i\in[r]}\theta^{2}_i\geq C_2\sqrt{\lambda/n}$ for $C_2$ large enough, and the final inequality follows from~\eqref{Eqn:RSSjMuSmall} for $n$ large enough. 
	
	Hence, for~\eqref{Eqn:RSS>Ln} to hold, we necessarily have
    \begin{equation}\label{Eqn:RSS>SymDiff}
    \frac{1}{n}|\hat{R}_j \triangle \tilde{R}_j|_{\mathcal{D}}\leq\frac{2\mathrm{RSS}_j(\mu)}{n\theta^2},
\end{equation}
which implies
    \begin{equation}\label{Eq:LossTildeR}
    L_n(\hat{R}_j, \tilde R_j)\leq\frac{1}{n}|\hat{R}_j \triangle \tilde{R}_j|_{\mathcal{D}}\leq\frac{2\mathrm{RSS}_j(\mu)}{n\theta^2} \leq \frac{12\sqrt{n\lambda}+32\lambda}{n\theta^2},
\end{equation}
where the final bound uses~\eqref{Eqn:RSSjMuSmall} again.
\textbf{Step 2:} Fixing $j\in[J_n]$ such that $\hat {R}_j\in \hat{\mathcal{R}}_{\mathrm{init}}$ as in Step 1. We proceed to show that there exists a unique $k\in[r]$ such that $L_n(\hat R_j,R_k)\leq L_n(\hat R_j,\tilde R_j)$.

Write $l_j:=|\{k\in[r]:\hat R_j\cap R_{k}\neq\emptyset\}|$. We claim that $\hat R_{j}\cap \tilde R_j\neq \emptyset$ and so $l_j\geq 1$. 
Otherwise, we have $\tilde R_j\subseteq \hat R_j\triangle \tilde R_j$ and $\hat{R}_j\triangle(B_j\setminus \tilde{R}_j)\subseteq B_j\setminus\tilde{R}_j$, and from Lemma~\ref{Lemma:CUSUM_signal}(d) and (e) again, we obtain that 
\begin{align*}
    \mathrm{RSS}_j(\mu)\geq\min\biggl\{\frac{\theta^2}{2}|\hat{R}_j\triangle(B_j\setminus \tilde{R}_j)|_{\mathcal{D}},{(\mathcal{T}_j(\mu))^2}\biggr\}.
\end{align*}
However,~\eqref{Eqn:RSSsq<T} and~\eqref{Eqn:SymDifLarge} together give a contradiction to the above inequality.

We further claim that $l_j\leq 1$. Otherwise, assuming $l_j\geq 2$, by Lemma~\ref{Lemma:ChompingBiscuits}, there exists a disc $A$ of radius at least $\omega/2$ such that $A\subseteq\hat{R}_j\setminus\tilde{R}_j$, and on the event $\mathcal{E}_4$, we have
\begin{equation}\label{Eq:Thm3Step1Hoeffding}
	{|\hat{R}_j\setminus\tilde{R}_j|_D} \geq nP_X(A)-7\sqrt{n\lambda}\geq n\underline{f}_XC_d(\omega/2)^{d}-7\sqrt{n\lambda},
\end{equation}
for some universal constant $C_d$. From~\eqref{Eqn:RSS>SymDiff}, we have
\begin{equation}
	\label{Eq:RSS>Ln}
	\mathrm{RSS}_j(\mu)\geq\frac{\theta^2}{2}|\hat{R}_j\setminus\tilde{R}_j|_{\mathcal{D}}\geq\frac{\theta^2}{2}\underline{f}_XC_{d}(\omega/2)^{d}n -\frac{7\theta^2}{2}\sqrt{n\lambda},
\end{equation}
 which contradicts with~\eqref{Eqn:RSSjMuSmall} by a similar calculation as in~\eqref{Eqn:SymDifLarge}. Therefore, we must have $l_j=1$, i.e.\ $\hat{R}_j$ intersects exactly one true change region.

Denote $B^{\mathrm{int}}_j:=\{x\in B_j:\mathrm{dist}(x,B^c_j)\geq\omega/2\}$. We further claim that, for $k\in[r]$, either $R_k\cap B_j^{\mathrm{int}}=\emptyset$ or $R_k\cap B_j^c=\emptyset$ holds, for otherwise Lemma~\ref{Lemma:ChompingBiscuits} implies that there is a disc $A'$ of radius at least $\omega/2$ such that $A'\subseteq R_k\cap(B_j\setminus B_j^{\mathrm{int}})\subseteq \tilde{R}_j\setminus \hat R_j$, and we have $|\tilde{R}_j\setminus\hat{R}_j|_{\mathcal{D}}\geq \underline{f}_XC_d(\omega/2)^dn-7\sqrt{n
\lambda}$ for some constant $C_d$. Again, by~\eqref{Eq:RSS>Ln} and a similar calculation as in~\eqref{Eqn:SymDifLarge}, we get a contradiction with~\eqref{Eqn:RSSjMuSmall}.

Combining the above, we have shown that there exists a unique $k\in[r]$ such that $\hat R_j\cap R_k\neq\emptyset$, and that $R_k\subseteq B_j$. In particular, we have $\hat{R}_j\setminus R_k=\hat{R}_j\setminus\tilde{R}$ and $R_k = R_k\cap B_j \subseteq \tilde R_j$. Therefore, $\hat{R}_j$ estimates an unique change region $R_k$, and 
\begin{align}
	L_n(\hat R_j, R_k) &= \frac{1}{n}\bigl|(\hat R_j \setminus R_k)\cup(R_k\setminus \hat R_j)\bigr|_{\mathcal{D}} \leq \frac{1}{n}\bigl|(\hat R_j\setminus \tilde R_j)\cup (\tilde R_j \setminus \hat R_j)\bigr|_{\mathcal{D}}  \nonumber\\
	&= L_n(\hat R_j, \tilde R_j)\leq \frac{12\sqrt{n\lambda}+32\lambda}{n\theta^2}.\label{Eq:LossRhatRtilde}
\end{align}

    \textbf{Step 3:} In this step, we show that every change region is estimated by one of the $J_n$ discs, and the corresponding disc is selected into $\hat{\mathcal{R}}_{\mathrm{init}}$.
    
    Fix $k\in[r]$, on the event $\mathcal{E}_{1}$, there exists $j\in[J_n] $ such that $\mathrm{dist}(R_k,B_j^{\mathrm{c}})\geq\omega/2,R_k \subseteq B_j, R_i \cap B_j = \emptyset, \forall i\neq k$. The dependency of $j$ on $k$ is suppressed for notation simplicity.
    
    We proceed to show that $\mathcal{T}_j(Y)>\lambda_T$. On the event $\mathcal{E}_{2}$, we have
	\begin{align}\label{Eq:CUSUMDataLowerBound}
		\mathcal{T}_{j}(Y) &\geq \mathcal{T}^{B_j}_{R_k}(Y) 
		\geq \mathcal{T}^{B_j}_{R_k}(\mu) - \sqrt{\lambda} \geq \sqrt{\min(\lvert R_k \rvert_{\mathcal{D}}, \lvert B_j \backslash R_k \rvert_\mathcal{D})}\frac{\theta_k}{\sqrt{2}} - \sqrt{\lambda},
	\end{align}
where the last inequality is due to Lemma \ref{Lemma:CUSUM_signal}(d). Denote by $\mathrm{Disc}(a,b)$ a disc on $\mathbb{S}^{d-1}$ with center $a\in\mathbb{S}^{d-1}$ and radius $b\in(0,\pi/2)$, then on the event $\mathcal{E}_4$, we have
\begin{align*}
P_X(B_j\setminus R_k) &\geq P_X(\mathrm{Disc}(\mathrm{Ctr}(R_k),\delta+\omega/2)) - P_X(\mathrm{Disc}(\mathrm{Ctr}(R_k),\delta)) \geq \underline{f}_XC_d(\omega/2)^{d}n,
\end{align*}
Thus, from~\eqref{Eq:CUSUMDataLowerBound} we have for $n$ sufficiently large that 
\[
\mathcal{T}_j(Y)> \frac{\theta_k}{\sqrt{2}}\Bigl\{\underline{f}_X nC_d\min\{(\omega/2)^{d},\delta^d\}-7\sqrt{n\lambda}\Bigr\}^{1/2} - \sqrt{\lambda} \stackrel{\mathrm{(i)}}\geq 5(n\lambda)^{1/4}-\sqrt{\lambda} > \lambda_T,
\]
where (i) follows from the assumption that $\delta\geq\omega$ and a similar argument as that in~\eqref{Eqn:SymDifLarge}. 

Lastly, we show that $\mathrm{RSS}_{j}(Y)<\gamma_j$. Observe that for $A\subset B_j$, it holds that $\{\mathcal{T}_{A}^{B_j}(Y)\}^2 = \mathrm{RSS}^{B_j}_{\emptyset}(Y)-\mathrm{RSS}^{B_j}_{A}(Y)$,
which implies that $\mathrm{RSS}_{j}(Y)\leq\mathrm{RSS}_{R_k}^{B_j}(Y)$ because $
\hat{R}_j$ is chosen in such a way that $\{\mathcal{T}_{j}(Y)\}^2\geq\{\mathcal{T}^{B_j}_{R_k}(Y)\}^2$. By the fact that $\mathrm{RSS}^{B_j}_{R_k}(\mu)=0$, on event $\mathcal{E}_3$ we have $\mathrm{RSS}^{B_j}_{R_k}(Y)\leq m_{{R}_k,B_j}+2\sqrt{\lambda m_{R_k,B_j}}+2\lambda=\gamma_j$ since $m_{R_k,B_j} = |B_j|_{\mathcal{D}}-2=m_j$. In conclusion, we have $\hat{R}_j\in\hat{\mathcal{R}}_{\mathrm{init}}$.
\textbf{Step 4:} Step 2 defines a map $h:\hat{\mathcal{R}}_{\mathrm{init}} \to [r]$ such that each $\hat R\in\hat{\mathcal{R}}_{\mathrm{init}}$ is close (in the sense of~\eqref{Eq:LossRhatRtilde}) to $R_{h(\hat R)}$. Step 2 implies that this map is a surjection. We claim that for any $\hat{R}\in \hat{\mathcal{R}}_{\mathrm{init}}$, we have $\sup_{x\in\hat{R}} \mathrm{dist}(x, R_{h(\hat{R})}) \leq \omega/2$, for otherwise, there exists a point $x_0\in\hat{R}\setminus R_{h(\hat{R})}$ such that $\mathrm{dist}(x_0,\hat{R}^{\mathrm{c}})\geq\omega/4$ and $\mathrm{dist}(x_0,R_{h(\hat{R})})\geq\omega/4$, thus $\hat{R} \setminus R_{h(\hat{R})}$ contains a disc $A$ of radius $\omega/4$ centered at $x_0$, which implies on event $\mathcal{E}_4$ that for sufficiently large $n$, 
\[  
\frac{1}{n}|\hat{R} \triangle {R}_{h(\hat{R})}|_{\mathcal{D}} \geq \underline{f}_XC_d(\omega/4)^{d} - 7\sqrt{\lambda/n} \geq \frac{50\sqrt{n\lambda}}{\theta^2} > \frac{12\sqrt{n\lambda}+32\lambda}{n\theta^2},
\]
contradicting~\eqref{Eq:LossTildeR}. Here, the calculation in the penultimate inequality is similar to that in~\eqref{Eqn:SymDifLarge}, with a possibly larger choice of $C_1$.  From this claim, we see that for any $\hat{R}, \hat{R}'\in \hat{\mathcal{R}}_{\mathrm{init}}$ such that $h(\hat{R})\neq h(\hat{R}')$, we must have by the fact that $d(R_{h(\hat{R})}, R_{h(\hat{R}')}) > \omega$ and the triangle inequality that $\hat{R}\cap \hat{R}' = \emptyset$. 

We further claim that for any $\hat{R},\hat{R}'\in\hat{\mathcal{R}}_{\mathrm{init}}$ such that $h(\hat{R})=h(\hat{R}')$, we have $\hat{R}\cap\hat{R}'\neq\emptyset$. Otherwise, we would have either $\sup_{x\in \hat{R}\cap R_{h(\hat{R})}}\mathrm{dist}(x,R_{h(\hat{R})})<\delta$ or $\sup_{x\in \hat{R}' \cap R_{h(\hat{R})}}\mathrm{dist}(x,R_{h(\hat{R})})<\delta$, and there would exist a disc of diameter $\delta$ in either $R_{h(\hat{R})}\setminus\hat{R}$ or $R_{h(\hat{R})}\setminus\hat{R}'$, and on the event $\mathcal{E}_4$, we have 
\[
\max\{L_n(\hat{R},R_{h(\hat{R})}),L_n(\hat{R}',R_{h(\hat{R})})\} \geq \underline{f}_XC_d(\delta/2)^d-7\sqrt{\lambda/n},
\]
violating~\eqref{Eq:LossTildeR} again by a similar calculation as in~\eqref{Eqn:SymDifLarge}. 

Therefore, by the construction of $\hat{\mathcal{R}}$ from $\hat{\mathcal{R}}_{\mathrm{init}}$, exactly one element of each of $h^{-1}(i)$ for $i\in[r]$ will be kept in $\hat{\mathcal{R}}$, i.e. the restriction of $h$ on $\hat{\mathcal{R}}$ defines a bijection to $[r]$. In particular, this implies that $\hat{r} = r$ as desired. We complete the proof by reminding ourselves that for each estimated change region $\hat{R}$ in $\hat{\mathcal{R}}$, we have the desired upper bound on the loss between $\hat R$ and $R_{h(\hat{R})}$ from~\eqref{Eq:LossRhatRtilde}. 
\end{proof}

\section{Ancillary Results}\label{appB}

\subsection{A minimax result}\label{appB1}
The following proposition establishes the minimax lower bound for estimating a single disc-shaped region of change $R \in\mathcal{S}$ on the sphere $\mathbb{S}^{d-1}$, based on $n$ independent observations $(X_i, Y_i)$, where $X_1,\ldots, X_n$ are uniformly distributed on $\mathbb{S}^{d-1}$, and $Y_1,\ldots,Y_n$ are generated according to \eqref{eqn:dataset}. Recall the loss function $L(\cdot, \cdot)$ defined in~\eqref{Eq:PopLoss}.
\begin{proposition}
\label{Prop:LowerBound}
	We have
	\begin{equation}
		\inf_{\hat{R}}\sup_{R\in\mathcal{S}} \mathbb{E}(L(\hat R, R)) \geq \frac{Cd\sigma^2 }{n\theta^2},
	\end{equation}
    for some universal constant $C>0$, where the infimum is taken over all estimators $\hat R$ for the change region, using data $(X_1,Y_1),\ldots,(X_n,Y_n)$ generated from the true model with a single change region $R$.
\end{proposition}
\begin{proof}
First assume that $d < 30$,  we can find two discs $D$ and $D'$ such that $L(D, D') = \sigma^2/(n\theta^2)$. Define $P_D$ to be the joint distribution of $(X_1,Y_1),\ldots,(X_n,Y_n)$ such that $Y_i\mid X_i \sim N(\theta\mathbbm{1}_{X_i\in D}, \sigma^2)$. We have 
	\begin{equation}\label{Eq:KL}
	d_{\mathrm{KL}}(P_D\,\|\, P_{D'}) = \mathbb{E}\biggl(\bigl|\{i: X_i\in D\triangle D'\}\bigr| \frac{\theta^2}{2\sigma^2}\biggr) =\frac{n\theta^2}{2\sigma^2}L(D', D) = \frac{1}{2}.
	\end{equation}

Then, by Le Cam's two-point lemma \citep{yu1997assouad} and Pinsker's inequality, we have 
\[  
\inf_{\hat{D}}\sup_{D} \mathbb{E}_{P_D}(L(\hat D, D)) \geq \frac{L(D,D')}{2}(1 - d_{\mathrm{TV}}(P_D, P_{D'})) \geq \frac{L(D,D')}{4} = \frac{\sigma^2}{4n\theta^2}\geq \frac{d\sigma^2}{120 n\theta^2},
\]
as desired.

Now, assume that $d \geq 30$.
By Gilbert--Varshamov lemma~\citep{gilbert1952,varshamov1957estimate}, there exists a subset $S$ of $\{0,1\}^d$ of cardinality $M \geq e^{\gamma d}$ for $\gamma = 0.096$, such that for every distinct pair of $u, v\in S$, we have $\|v\|_0 = \lceil d/2\rceil$ and $\|v-u\|_0 \geq \lceil d / 2\rceil $.

We define a set of discs $\mathcal{D}$ centred at $v/\lVert v\rVert_2$ for $v\in S$ with radius $r$ chosen so that for $D\in\mathcal{D}$, we have 
\[  
\mu(D) = \frac{\gamma d\sigma^2}{2 n\theta^2}\mu(\mathbb{S}^{d-1}),
\]
where $\mu$ is the Lebesgue measure on the sphere. For $n$ sufficiently large, we have the radius $r$ is smaller than $1/2$. On the other hand, by construction, for $u,v \in S$, we have $\|u/\|u\|_2 - v/\|v\|_2\|_2 \geq 1$. Hence all discs in $\mathcal{D}$ are disjoint. In particular, we have  
    \[  
    L(D', D) = \frac{2\mu(D)}{\mu(\mathbb{S}^{d-1})} = \frac{\gamma d\sigma^2}{n\theta^2}.
    \]
	
	By a similar calculation as in~\eqref{Eq:KL}, we have 
	\[
	d_{\mathrm{KL}}(P_D\,\|\, P_{D'})  =\frac{n\theta^2}{2\sigma^2}L(D', D) \leq \frac{1}{2}\gamma d.
	\]
	Finally, by Fano's lemma \citep{yu1997assouad}, we have 
	\begin{align*}
	\inf_{\hat{D}}\sup_{D} \mathbb{E}_{P_D}(L(\hat D, D)) &\geq \frac{\min_{D\neq D'\in\mathcal{D}} L(D,D')}{2} \biggl(1 - \frac{\max_{D\neq D'\in\mathcal{D}} d_{\mathrm{KL}}(P_D\,\|\,P_{D'})+\log 2}{\log |\mathcal{D}|}\biggr)\\
    &\geq \frac{\gamma d\sigma^2}{2n\theta^2}\biggl(1 - \frac{\gamma d / 2+\log 2}{\gamma d}\biggr) \geq \frac{\gamma d\sigma^2}{8n\theta^2},
	\end{align*}
    as desired. 
\end{proof}

\subsection{Other lemmas used in the proofs}
\begin{lemma}\label{Lemma:vcdOperationComplement}
	Let $\Omega$ denote a set and $\mathcal{A}$ denote a family of subsets from $\Omega$ such that $\text{VCD}(\mathcal{A})=k$. Let $\mathcal{B}:=\{A^{\text{c}}:A\in\mathcal{A}\}$. Then any set shattered by $\mathcal{A}$ is also shattered by $\mathcal{B}$, and $\text{VCD}(\mathcal{B})=k$.
\end{lemma}
\begin{proof}
	Since $\text{VCD}(\mathcal{A})=k$, there exists a set $S\subset\Omega$ such that $|S|=k$, and $\{S\cap A:A\in \mathcal{A}\}=2^{S}$. Let $U\in 2^{S}$, we have $S\setminus U\in 2^S$, then there exists $R_1\in\mathcal{A}$ such that $S\setminus U=S\cap R_1$. Therefore, $U=S\cap R_1^{\text{c}}$. Since $U$ is arbitrarily taken from $2^S$, we have $2^{S}=\{S\cap A:A\in \mathcal{B}\}$ and $\text{VCD}(\mathcal{B})\geq k$. Similarly, for any set that cannot be shattered by $\mathcal{A}$, it cannot be shattered by $\mathcal{B}$ neither.
\end{proof}
\begin{lemma}\label{Lemma:vcdOperationDiffCap}	Let $\mathcal{M}$ be a generic set and $\mathcal{A}$ denote a family of subsets from $\mathcal{M}$. Let $R\subseteq\mathcal{M}$, and define $\mathcal{A}\cap_*R:=\{A\cap R:A\in\mathcal{A}\},\mathcal{A}\setminus_*R:=\{A\setminus R:A\in\mathcal{A}\}$ and $R\setminus_*\mathcal{A}:=\{R\setminus A:A\in\mathcal{A}\}$. Then we have the following holds
\renewcommand{\labelenumi}{(\alph{enumi})}
	\begin{enumerate}
		\item $\mathrm{VCD}(\mathcal{A}\cap_* R)\leq \min\{\mathrm{VCD}(\mathcal{A}), |R|\}$;
		\item $\mathrm{VCD}(\mathcal{A}\setminus_* R)\leq \min\{\mathrm{VCD}(\mathcal{A}), |R^{\mathrm{c}}|\}$;
		\item $\mathrm{VCD}(R\setminus_*\mathcal{A})\leq \min\{\mathrm{VCD}(\mathcal{A}), |R|\}$.
	\end{enumerate} 
\end{lemma}
\begin{proof}
	To show result (a), let $S\subseteq\mathcal{M}$ be a set shattered by $\mathcal{A} \cap_{*}R$. Since $S\setminus R$ is a subset of $S$, there exists a set $A_0\in\mathcal{A}$ such that $S\setminus R=S\cap (R\cap A_0)$. This implies that $S\setminus R=\emptyset$. Therefore, we have $S\subseteq R$ and $S\cap R=S$, and $2^S=\{S\cap (R\cap A):A\in\mathcal{A}\}=\{S\cap A:A\in\mathcal{A}\}$, which further implies that $\mathcal{A}$ shatters $S$ as well, and $\mathrm{VCD}(\mathcal{A}\cap_* R)\leq\mathrm{VCD}(\mathcal{A})$. For any $A\in\mathcal{A}$, we have $|R\cap A|\leq |R|$, hence $\mathcal{A}\cap_{*}R$ cannot shatter any set whose cardinality is larger than $|R|$.
	
	Noting that $\mathcal{A}\setminus_* R = \mathcal{A}\cap_* R^{\mathrm{c}}$, (b) follows immediately from (a).
	
	To demonstrate (c), consider $\mathcal{B}:=\{A^{\mathrm{c}}:A\in\mathcal{A}\}$, we have $\mathrm{VCD}(\mathcal{B})=\mathrm{VCD}(\mathcal{A})$ by Lemma~\ref{Lemma:vcdOperationComplement}. And (c) follows from (a) by noticing that $R\setminus_{*}\mathcal{A}=\mathcal{B}\cap_{*}R$.
\end{proof}
\begin{lemma}\label{Lemma:vcdOperationSymDiff}	Let $\Omega$ be a generic set and $\mathcal{A}$ denote a family of subsets from $\Omega$. Let $R\subseteq\Omega$. Then we have $\mathrm{VCD}(\{A\triangle R:A\in\mathcal{A}\})=\mathrm{VCD}(\mathcal{A})$.
\end{lemma}
\begin{proof}
	Suppose $\mathrm{VCD}(\{A\triangle R:A\in\mathcal{A}\})=\mathrm{VCD}(\mathcal{A})+q$ for some $q\in\mathbb{Z}$. Then we have $\mathrm{VCD}(\{(A\triangle R)\triangle R:A\in\mathcal{A}\})=\mathrm{VCD}(\mathcal{A})+2q$.
	Since for $A\in\mathcal{A}$,
	\begin{align*}
		(A\triangle R)\triangle R &=((A\cup R)\setminus(A\cap R))\triangle R \\
		&= \{((A\cup R)\setminus(A\cap R))\cup R\} \setminus\{((A\cup R)\setminus(A\cap R))\cap R\}\\
		&=(A\cup R)\setminus(R\setminus A)=A,
	\end{align*}
    we have
	\begin{equation*}
		\mathrm{VCD}(\{(A\triangle R)\triangle R:A\in\mathcal{A}\})=\mathrm{VCD}(\mathcal{A}).
	\end{equation*}
	Therefore, we have $q=0$, and the result follows.
\end{proof}
\begin{lemma}\label{lemma:vcdim_operation_union}
	Suppose $\mathcal{A}_1,\mathcal{A}_2$ are families of subsets of a set $\Omega$ such that $\text{VCD}(\mathcal{A}_1)=k_1$ and $\text{VCD}(\mathcal{A}_2)=k_2$. Assume $|\Omega|\geq k_1+k_2+2$. Then $\text{VCD}(\mathcal{A}_1\cup\mathcal{A}_2)\leq k_1+k_2+1$.
\end{lemma}
\begin{proof}
	Let $S_1\subseteq\Omega$ such that $|S_1|=k_1+1$. We can find a set $S_2\subseteq\Omega\setminus S_1$ such that $|S_2|= k_2+1$. We have $\{S_i\cap A:A\in\mathcal{A}_i\}\neq2^{S_i},i=1,2$. Let $U_i\in 2^{S_i}\setminus\{S_i\cap A:A\in\mathcal{A}_i\},i=1,2$, and denote $U=U_1\cup U_2$. We have $U\subset S_1\cup S_2$ and $|S_1\cup S_2|=k_1+k_2+2$. 

    We show that $S_1\cup S_2$ cannot be shattered by $\mathcal{A}_1\cup\mathcal{A}_2$, and the proof is complete since the choices of $S_1$ and $S_2$ are arbitrary.
    To this end, suppose there exists $B\in\mathcal{A}_1\cup\mathcal{A}_2$ such that $U=B\cap (S_1\cup S_2)$, then it must hold that $U_i=B\cap S_i,i=1,2$, since $S_1\cap S_2=\emptyset$. As $B$ is taken either from $\mathcal{A}_1$ or $\mathcal{A}_2$, we have $U_i\in S\cap A_i,i=1,2$, which contradicts with the definition of $U_i,i=1,2$. Therefore, $U$ is a set that cannot be picked out from $S_1\cup S_2$ via set intersection by any set from $\mathcal{A}_1\cup\mathcal{A}_2$. 
\end{proof}
\begin{lemma}\label{lemma:vcdim_operation_union_complement}Suppose $\mathcal{A}$ is a family of subsets of a set $\Omega$ with $\text{VCD}(\mathcal{A})=k$. Assume $\Omega\setminus \cup_{A\in\mathcal{A}}A\neq \emptyset$. Let $\mathcal{B}=\{A^{\text{c}}:A\in\mathcal{A}\}$. Then $\text{VCD}(\mathcal{A}\cup\mathcal{B})\geq k+1$.
\end{lemma}
\begin{proof}
	Let $S_0\subseteq\Omega$ be a set such that $|S_0|=k$ and $S_0$ is shattered by $\mathcal{A}$. Then $S_0$ is also shattered by $\mathcal{B}$ by Lemma \ref{Lemma:vcdOperationComplement}. Let $x_1\in\Omega\setminus S_0$ and denote $S_1:=S_0\cup\{x_1\}$.
    
    We proceed to show that $S_1$ is shattered by $\mathcal{A}\cup\mathcal{B}$ to complete the proof. It is equivalent to show that all members in the power set of $S_1$ can be picked out from $S_1$ by $\mathcal{A}\cup\mathcal{B}$, where we say a set $C_1$ can be picked out from a set $C_2$ by a family of sets $\mathcal{C}$ if there exists a set $C_3\in\mathcal{C}$ so that $C_1=C_2\cap C_3$. The power set of $S_1$ can be decomposed as $2^{S_1}=2^{S_0}\cup \mathcal{I}$ where each set in $\mathcal{I}$ is a union of $\{x_1\}$ and a member in $2^{S_0}$. For any $U\in2^{S_0}$, there exist $A_U\in \mathcal{A}$ such that $A_U\cap S_0=U$, and by specifically selecting
	$x_1 \in\Omega\setminus \bigcup_{U'\in 2^{S_0}}A_{U'}$, we know $U$ can also be picked out from $S_1$ by $A_U$ since $x_1\not\in A_{U}$. When defining the set $\bigcup_{U'\in 2^{S_0}}A_{U'}$, in case $A_{U'}$ is not unique for some $U'$ in $2^{S_{0}}$, we may choose $A_{U'}$ to be the one that has the largest intersection with the union of chosen ones. Members in $\mathcal{I}$ can also be picked out by $\mathcal{A}\cup\mathcal{B}$, since $S_1\setminus I\in 2^{S_0}$ for $I\in\mathcal{I}$, and for any $C\subseteq S_1$, if $S_1\setminus C$ can be picked out from $S_1$ by $\mathcal{A}\cup\mathcal{B}$ then $C$ can also be picked out from $S_1$ by $\mathcal{A}\cup\mathcal{B}$.
\end{proof}
\begin{lemma}
	\label{lemma:vcd_disk_lb}
	Let $\mathcal{A}_{d-1}$ denote the collection of discs on $\mathbb{S}^{d-1}$. We have $\text{VCD}(\mathcal{A}_{d-1})\geq d+1$.
\end{lemma}
\begin{proof}
	Let $e:=\{\bm{e}_{1},\dots,\bm{e}_{d+1}\}\in\mathbb{R}^{d+1}$ be a standard basis of $\mathbb{R}^{d+1}$. We know $e$ can be shattered by the family of hyperplanes $\mathbb{H}^{d}\subseteq 2^{\mathbb{R}^{d+1}}$. Let $U\subset e$, then there is some hyperplane $H\in\mathbb{H}^{d}$ whose positive halfspace separates $U$ from $e\setminus U$. In $\mathbb{R}^{d+1}$, $e$ determines a circle $S_e$, and the embedding of $S_e$ into $\mathbb{R}^{d}$ is equivalent to $\mathbb{S}^{d-1}$. Since the positive halfspace of $H$ determines a disc $A_H\in\mathcal{A}_{d-1}$ so that $A_H$ picks out $U$ from $\mathbb{S}^{d-1}$, and $U$ is arbitrarily chosen, we know there exists a set of cardinality $d+1$ that can be shattered by $\mathcal{A}_{d-1}$. Therefore, $\text{VCD}(\mathcal{A}_{d-1})\geq d+1$.
\end{proof}
\begin{lemma}\label{lemma:vcdim_discs}
	Let $\mathcal{A}_{d-1}$ denote the family of discs on $\mathbb{S}^{d-1}$. Then $\text{VCD}(\mathcal{A}_{d-1})= d+1$.
\end{lemma}
\begin{proof}  
	Let $x_0\in\mathbb{S}^{d-1}$, and let $\Pi_0(\bm{x}):\mathbb{S}^{d-1}\rightarrow\mathbb{R}^{d-1}$ be a mapping defined by setting $\Pi_0(-x_0):=\bm{0}$, and $\Pi_0(x):=\tan(\angle(-x_0,x))$ where for $x\in\mathbb{S}^{d-1}$, $\angle(-x_0,x)$ is the angle formed between $-x_0$ and $x$. We leave $\Pi_0(x_0)$ undefined, and it will not affect the result.
    
    We first show that $\text{VCD}(\mathcal{A}_{d-1})\geq d+1$. Let $\mathbb{B}^{d-1}:=\{B_{\bm{a},\bm{b},\bm{c}}:\bm{a},\bm{b},\bm{c}\in\mathbb{R}^{d-1}\}$ denote the collection of balls where $B_{\bm{a},\bm{b},\bm{c}}:=\{\bm{x}\in\mathbb{R}^{d-1}:||\bm{a}^{\text{T}}\bm{x}-\bm{b}||^2\leq \bm{c}\}$ denotes a ball for given $\bm{a},\bm{b},\bm{c}\in\mathbb{R}^{d-1}$, and let $\mathcal{B}_{d-1}:=\{B\subset\mathbb{R}^{d-1}:B\in\mathbb{B}^{d-1} \text{ or }B^{\text{c}}\in\mathbb{B}^{d-1}\}$ denote the collection of both balls and complements of balls. We say a set $C_1$ can be picked out from a set $C_2$ by a family of sets $\mathcal{C}$ if there exists a set $C_3\in\mathcal{C}$ so that $C_1=C_2\cap C_3$. If $S\subseteq\mathbb{S}^{d-1}$ can be picked out by $A_{S}\in\mathcal{A}_{d-1}$ from $\mathbb{S}^{d-1}$, we know $\Pi_0(S)$ can be picked out from $\mathbb{R}^{d-1}$ by $\Pi_0(A_S)\in\mathcal{B}_{d-1}$. Reversely, if $S\subseteq\mathbb{R}^{d-1}$ can be picked out by $B_S\in\mathcal{B}_{d-1}$, then $\Pi_0^{-1}(S)$ can be picked out by $\Pi_0^{-1}(B_S)$. Hence $\text{VCD}(\mathcal{A}_{d-1})=\text{VCD}(\mathcal{B}_{d-1})$.	Let $\mathcal{F}:=\{f(\bm{x})=R_1||\bm{x}||^2+\sum_{i=1}^{d-1}a_ix_i +a_{d}:\bm{a}=(R_1,\dots,a_{d})\in\mathbb{R}^{d+1}\}$ be a family of functions. Let $B:=\{\bm{x}:||\bm{x}-\bm{a}||\leq b\}$ be a ball where $\bm{a}\in\mathbb{R}^{d-1},b\in\mathbb{R}$. Observe that $B$ corresponds to the sublevel set of the function $f(\bm{x}):=||\bm{x}-\bm{a}||-b$ at $0$, and $B^{c}$ corresponds to the sublevel set of the function $g(\bm{x})=-||\bm{x}-\bm{a}||+b$ at $0$. Since $B$ is arbitrarily chosen, any set shattered by $\mathcal{B}_{d-1}$ corresponds to a set shattered by $\mathcal{F}$. Hence $\text{VCD}(\mathcal{A}_{d-1})=\text{VCD}(\mathcal{B}_{d-1})\leq\text{VCD}(\mathcal{F})\leq d+1$, as $\mathcal{F}$ has dimension $d+1$ as a vector space.

We complete the proof by showing that $\text{VCD}(\mathcal{A}_{d-1})\geq d+1$. This follows directly from Lemma~\ref{lemma:vcd_disk_lb}, alternatively, we have $\text{VCD}(\mathcal{A}_{d-1})=\text{VCD}(\mathcal{B}_{d-1})\geq\text{VCD}(\mathbb{B}^{d-1})+1=d+1$ where the inequality follows from Lemma \ref{lemma:vcdim_operation_union_complement}.
\end{proof}

\begin{lemma}
	\label{lemma:param_est_lb}
Let $\mu$ be the Lebesgue measure on sphere $\mathbb{S}^{d-1}$. For arbitrary $\alpha, \alpha'\in\mathbb{S}^{d-1},\beta,\beta'\in[0,1]$, there exists a constant $C_{d,\beta}$, depending only on $d$ and $\beta$, such that
\[
	\|\alpha-\alpha^{\prime} \|+\left|\beta-\beta^{\prime}\right| \leq C_{d,\beta} \mu\left(A_{\alpha, \beta} \triangle A_{\alpha', \beta'}\right) .
\]
\end{lemma}
\begin{proof}
We start by bounding $\left\|\alpha-\alpha'\right\|$ by $\mu(A_{\alpha,\beta}\triangle A_{\alpha',\beta'})$. Since we have $\mu(A_{\alpha,\beta}\setminus A_{\alpha',\beta'})\geq \mu(A_{\alpha,\beta}\setminus A_{\alpha',\beta})$ when $\beta\leq\beta'$ and $\mu(A_{\alpha',\beta'}\setminus A_{\alpha,\beta})\geq \mu(A_{\alpha',\beta}\setminus A_{\alpha,\beta})$ when $\beta\geq\beta'$, it holds that
\begin{align}
	\label{eqn:proof_lemma_disc_para_lb_1}
	\mu\left(A_{\alpha, \beta} \triangle A_{\alpha', \beta'}\right) \geq \frac{1}{2} \mu\left(A_{\alpha, \beta} \triangle A_{ \alpha', \beta}\right),
\end{align}
and it reduces to find a lower bound of the symmetric difference of disc $A_{\alpha,\beta}$ and disc $A_{\alpha',\beta}$ that have the same intercept. To this end, we consider two cases when $\alpha^\top\alpha'\geq\beta$ and $\alpha^\top\alpha'<\beta$ separately.

Assume $\alpha^\top\alpha'\geq\beta$. We define helper quantities with an example when $d=3$ visualised in Figure~\ref{fig:proofLemmaGeometry}. Let $\theta=\angle(\alpha,\alpha')$ denote the angle between the two centres and let $L:=\left\{x \in \mathbb{R}^{d}: x^{\top} \alpha=x^{\top} \alpha'=\beta,\|x\|\leq 1\right\}$ be the line segment of intersection between hyperplanes $\left\{x^{\top} \alpha=\beta\right\}$ and $\left\{x^{\top} \alpha'=\beta\right\}$. Let $Q \in L$ be the midpoint of $L$ such that $Q^{\top} \alpha=\beta, Q^{\top} \alpha^{\prime}=\beta$ and $\left\|Q-x_l\right\|=\left\|Q-x_r\right\|$ where $\left\{x_l, x_r\right\}=L \cap \mathbb{S}^{d-1}$ are the boundary points of $L$. Specifically, we have 
\begin{equation}
    \label{Eq:proofGeomtryQ}
    Q=\frac{\beta}{1+\cos(\theta)}\left(\alpha+\alpha^{\prime}\right).
\end{equation}
		 Let $R$ denote the point on the circle $\left\{x^{\top} \alpha=\beta\right\} \cap \mathbb{S}^{d-1}$ that lies on the same line with $Q$ and $\beta\alpha$ (the center of the circle $\left\{x^{\top} \alpha=\beta\right\} \cap \mathbb{S}^{d-1}$), and $R$ is chosen to be close to $Q$ so that $R^{\top}Q/\|Q\|>R^\top\alpha$. Denote $\mathbb{S}(Q, r)$ the sphere centered at $Q$ with radius $r$ where $r:=\|R-Q\|_2$. For $x\in\mathbb{S}(Q,r)$, it can be decomposed as $x=Q+(x-Q) r$, and we let $\mathrm{p}(x)$ be the projection of $x$ onto $\mathbb{S}^{d-1}$ along the direction $x-Q$. Specifically, we have $\mathrm{p}(x)=Q+(x-Q) l_x$ where $l_x>0$ is chosen such that $\|\mathrm{p}(x)\|_2=1$.
         
\begin{figure}[ht]
\centering
\vspace{6mm} 
\begin{tikzpicture}[scale=0.75]
  \begin{scope}
    \draw[very thick] (0,0) circle (2.4);
    \draw[thick, red!70] (-0.9,0) circle (1.45);
    \draw[thick, red!70] (0.9,0) circle (1.45);
	\fill[red!70] (0.9,0) circle (1.6pt);
	\fill[red!70] (-0.9,0) circle (1.6pt);
	\node[red!70, left=2pt] at (-0.9,0) {$\alpha$};
    \node[red!70] at (1.15,0.35) {$\alpha'$};
    \draw[blue!70, line width=0.8pt, dashed] (0,1.13688) -- (0,-1.13688);
    \fill[blue!70] (0,1.13688) circle (1.6pt);
    \fill[blue!70] (0,-1.13688) circle (1.6pt);
    \node[blue!70, above=2pt] at (0,1.13688) {$x_r$};
    \node[blue!70, below=2pt] at (0,-1.13688) {$x_\ell$};
	\fill[blue!70] (0,0) circle (1.6pt);
    \node[blue!70] at (-0.25,0) {$Q$};
	\fill[blue!70] (0.55,0) circle (1.6pt);
    \node[blue!70] at (0.7,-0.3) {$R$};
    \node[below] at (0,-2.75) {(a)};
  \end{scope}

  \begin{scope}[xshift=6.4cm]
    \draw[very thick] (0,0) circle (2.4);
	\fill[red!70] (0,0) circle (1.6pt);
    \draw[thick, red!70] (-2.4,0) -- (1,2.18174);
    \draw[thick, red!70] (2.4,0) -- (-1,2.18174);
    \draw[red!70, dashed] (0,0) -- (-1.29615,2.01990);
    \draw[red!70, dashed] (0,0) -- (1.29615,2.01990);
	\fill[red!70] (-1.29615,2.01990) circle (1.6pt);
	\fill[red!70] (1.29615,2.01990) circle (1.6pt);
    \draw[red!70] (0,0)++(57.31:0.55) arc (57.31:122.69:0.55);
    \node[red!70] at (-1.4,2.3) {$\alpha$};
    \node[red!70] at (1.6,2.2) {$\alpha'$};
    \node[red!70] at (0,0.75) {$\theta$};
    \fill[blue!70] (0,1.54005) circle (1.6pt);
    \fill[blue!70] (1,2.18174) circle (1.6pt);
    \node[blue!70] at (0,1.2) {$Q$};
    \node[blue!70] at (1,2.5) {$R$};
	\draw[blue!70, decorate, decoration={brace, amplitude=4pt}] (0,1.54005) -- (1,2.18174)
      node[midway, above=1pt] {$r$};
    \node[below] at (0,-2.75) {(b)};
  \end{scope}
\end{tikzpicture}
\caption{Example illustration of the points $Q$ specified in~\eqref{Eq:proofGeomtryQ} and $R$, the angle $\theta$ and the line segment of length $r$ when $d=3$ with $A_{\alpha,\beta}$ and $A_{\alpha',\beta}$ shown as red discs. Panel (a) shows the projection of the sphere viewed from above. Panel (b) shows the projection of the sphere viewed from aside.}
\label{fig:proofLemmaGeometry}
\end{figure}
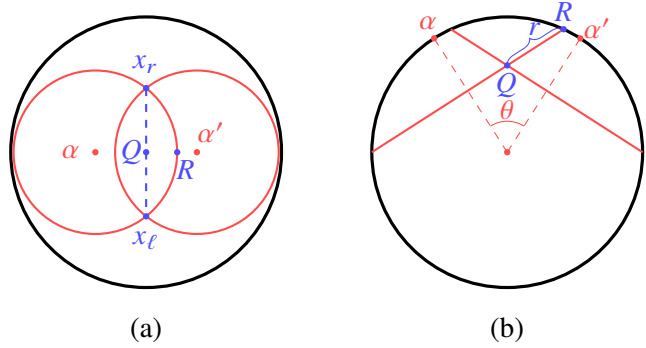

Let $B_{\alpha, \beta}:=\{x\in\mathbb{S}(Q,r):x^{\top}\alpha\geq\beta\}$ and $B_{\alpha',\beta}:=\{x\in\mathbb{S}(Q,r):x^{\top}\alpha'\geq\beta\}$ be two spherical caps on $\mathbb{S}(Q, r)$. For $x \in B_{\alpha, \beta} \setminus B_{\alpha', \beta}$, we have
			\begin{align}\label{Eq:proofLemmaGeomProjection}
				\mathrm{p}(x)^{\top} \alpha & =Q^{\top} \alpha+{l_x}\left(x^{\top} \alpha-Q^{\top} \alpha\right)  =\beta+{l_x}\left(x^{\top} \alpha-\beta\right) \geqslant \beta.
			\end{align}
	Similarly, we have $\mathrm{p}(x)^{\top} \alpha^{\prime}=\beta+{l_x}\left(x^{\top} \alpha^{\prime}-\beta\right)<\beta$, together with~\eqref{Eq:proofLemmaGeomProjection}, we have $\mathrm{p}(x)\in A_{\alpha, \beta} \backslash A_{\alpha^{\prime}, \beta}$. By a similar argument, we also have $\mathrm{p}(y)\in A_{\alpha', \beta} \backslash A_{\alpha, \beta}$ for $y\in B_{\alpha', \beta} \setminus B_{\alpha, \beta}$.
    
    As the two sets $\{\mathrm{p}(x):x\in B_{\alpha,\beta}\triangle B_{\alpha',\beta}\}$ and $\{x:x\in B_{\alpha,\beta}\triangle B_{\alpha',\beta}\}$ have the same Lebesgue measure, we have for $C_d=\sqrt{2}\pi^{d/2-1}/\Gamma{(d/2)}$,
		\begin{align*}
			\mu\left(A_{\alpha, \beta} \triangle A_{\alpha^{\prime}, \beta}\right) \geq \mu(B _{\alpha', \beta }\triangle B_{ \alpha, \beta})=r^{d-1}\theta / \pi
			\frac{2 \pi^{d / 2}}{ \Gamma(d / 2)}\geq C_{d}r^{d-1}\|\alpha-\alpha'\|_2,
		\end{align*} where the last inequality follows from $\theta^2\geq\sin^2\theta=1-(\alpha^{\top}\alpha')^2\geq\|\alpha-\alpha'\|_2^2/2$, and the equality is due to the fact that for $\alpha\in\mathbb{S}^{d-1},\beta\in[0,\pi]$, the surface area of $A_{\alpha,\beta}$ is given by
	\begin{align}\label{Eq:surfaceArea}
		\mu(A_{\alpha,\beta})&= \frac{\pi^{d/2-1/2}}{\Gamma \left( \frac{d-1}{2} \right)} \int_0^{(1+\beta)(1-\beta)}t^{\frac{d}{2}-\frac{3}{2}} (1-t)^{-\frac{1}{2}} dt.
	\end{align}
        We proceed to argue that $r\geq\sqrt{1-\beta^2}/2$ for any $\alpha, \alpha^{\prime}\in\mathbb{S}^{d-1}, \beta\in[0,1]$, which consequently gives that for $C'_d=C_d/2^d$, we have
        \begin{align}
            \label{eqn:proof_lemma_disc_para_lb_2}
			\mu\left(A_{\alpha, \beta} \triangle A_{\alpha^{\prime}, \beta}\right) \geq C'_{d}{(1-\beta^2)^{(d-1)/2}}\|\alpha-\alpha'\|_2,
        \end{align}
        when $\alpha^\top\alpha'\geq\beta$. If $\alpha^{\top}\alpha'=\beta$, we have $R=\alpha',R^{\top}Q = \alpha^{\top}Q$, and $\|\beta\alpha\|_2\leq\|R\|_2$, hence $r = \|R-Q\|\geq\|R-\beta\alpha\|/2=\sqrt{1-\beta^2}/2$ by the law of cosine.
        If $\alpha^{\top} \alpha^{\prime} \geqslant \beta$, the centre $\alpha$ lies on the other disc $ A _{\alpha', \beta}$, and since the value of $\|Q-\beta\alpha\|=\tan(\theta/2)\beta$ decreases compared to its value when $\alpha^{\top}\alpha'=\beta$, we have $r \geq \sqrt{1-\beta^2}/2$.
	 
	 On the other hand, assume $\alpha^{\top} \alpha^{\prime}<\beta$. In this case we have $\alpha \notin A_{\alpha', \beta}$, and we bound $\mu(A_{\alpha,\beta}\triangle A_{\alpha',\beta})$ directly without projection onto $\mathbb{S}(Q,r)$. Let $\gamma\in\mathbb{S}^{d-1}$ be a point such that $\gamma^{\top}\alpha=0,\gamma^{\top}\alpha'\leq0$ and $\gamma$ lies on the same plane as $\alpha,\alpha',Q$ and $R$ in the sense that $\gamma^{\top}R=-\sqrt{1-\beta^2}$. For $x \in A_{\alpha,\beta}$, we can decompose $x=a\alpha+b\gamma+\eta$ where $a\in[\beta,1]$, $\eta^{\top}\alpha'=0$, and either $b\in[0,\sqrt{1-\beta^2}]$ or $b\in[-\sqrt{1-\beta^2},0]$ with each case holds for exactly half of the points in $A_{\alpha,\beta}$. For $x\in A_{\alpha,\beta}$ such that $b=x^{\top}\gamma\geq0$, it holds that $x^{\top}\alpha'=a\alpha^{\top}\alpha'+b\gamma^{\top}\alpha'+\eta^{\top}\alpha'\leq a\beta\leq \beta$, hence 
	 \begin{align*}
	 	\mu\left(A_{\alpha,\beta}\setminus A_{\alpha',\beta}\right) \geq \mu(A_{\alpha,\beta})/2= C_{d,\beta}\geq C_{d,\beta}/2||\alpha-\alpha'||_2,
	 \end{align*}
	 where 
		\[C_{d,\beta}=\frac{\pi^{d/2-1/2}}{2\Gamma \left( \frac{d-1}{2} \right)} \int_0^{(1+\beta)(1-\beta)}t^{\frac{d}{2}-\frac{3}{2}} (1-t)^{-\frac{1}{2}} dt.\]
        By a similar argument for $x\in A_{\alpha,\beta}$ such that $b=x^{\top}\gamma\leq0$, we have \begin{align}\label{eqn:proof_lemma_disc_para_lb_3}\mu\left(A_{\alpha',\beta}\triangle A_{\alpha,\beta}\right)  \geq C_{d,\beta}||\alpha-\alpha'||_2,
        \end{align}
         when $\alpha^\top\alpha'<\beta$.
         
		Combining (\ref{eqn:proof_lemma_disc_para_lb_1}), (\ref{eqn:proof_lemma_disc_para_lb_2}) and (\ref{eqn:proof_lemma_disc_para_lb_3}), we get
		\begin{align}\label{Eq:proofSymDiff>Alpha}
			\mu\left(A_{\alpha, \beta} \triangle A_{\alpha^{\prime}, \beta^{\prime}}\right) \geq C_{d,\beta}'\|\alpha-\alpha'\|_2,
		\end{align}
	where $C_{d,\beta}'=C_{d,\beta}\wedge C'_{d}{(1-\beta^2)^{(d-1)/2}}$.
	
	We now bound $|\beta-\beta'|$. For arbitrary $\alpha\in\mathbb{S}^{d-1},\beta\in[0,1]$, we have 
		\begin{align*}
			\mu(A_{\alpha,\beta})&= \frac{\pi^{d/2-1/2}}{\Gamma \left( \frac{d-1}{2} \right)} \int_0^{(1+\beta)(1-\beta)}t^{\frac{d}{2}-\frac{3}{2}} (1-t)^{-\frac{1}{2}} dt\\
            &\geq C_{d}''(1-\beta^2)^{(d-1)/2}\geq C_{d}''(1-\beta)^{d}\geq C_{d}''(1-d\beta),
		\end{align*}
	where $C_{d}''=2\pi^{d/2-1/2}/\Gamma \left( d/2-1/2 \right)/(d-1)$ due to the fact that $(1-t)^{-1/2}\geq1,t\in[0,1]$.
		Hence
		\begin{align}\label{Eq:proofSymDiff>Beta}
		\mu\left(A_{\alpha, \beta} \triangle A_{\alpha^{\prime}, \beta^{\prime}}\right) &=\mu(A_{\alpha,\beta}\cup A_{\alpha',\beta'})-\mu(A_{\alpha,\beta}\cap A_{\alpha',\beta'})\nonumber\\
        &\geq \mu\left(A_{\alpha,\beta\wedge\beta'}\right)-\mu(A_{\alpha,\beta\vee\beta'})\geq C_d''d|\beta-\beta'|.
		\end{align}

        Combining~\eqref{Eq:proofSymDiff>Alpha} and~\eqref{Eq:proofSymDiff>Beta}, we arrive at $$\mu\left(A_{\alpha, \beta} \triangle A_{\alpha^{\prime}, \beta^{\prime}}\right) \geq C_{d,\beta}'\|\alpha-\alpha'\|_2\vee C_d''d|\beta-\beta'|,$$ and
        $$ \mu\left(A_{\alpha, \beta} \triangle A_{\alpha^{\prime}, \beta^{\prime}}\right) \geq C_{d,\beta}''(\|\alpha-\alpha'\|_2+ |\beta-\beta'|),$$ for some constant $C_{d,\beta}''=C_{d,\beta}'/2\wedge C_{d}''d/2$, as desired.
\end{proof}
\begin{lemma}
\label{Lemma:ChompingBiscuits}
    For $r\geq 2$, let $A_1,\ldots,A_r$ be disjoint discs in a Riemannian manifold $\mathcal{M}$ whose pairwise distance is at least $2\omega$. If a disc $B$ intersects each of $A_j$ for $j\in[r]$, then there exists a disc $C$ of radius at least $\omega$ such that $C^{\circ} \subseteq B \setminus \cup_{j=1}^r {A_j}$, where $C^\circ$ denotes the interior of $C$.
\end{lemma}
\begin{proof}
    We first claim that $B\setminus\cup_{j=1}^{r}A_j\neq\emptyset$. Otherwise, assume $B\subseteq\cup_{j=1}^{r}A_j$. Because $B\cap A_j\neq\emptyset$ for $j\in[r]$, by the convexity of $B$, there exists a line segment $l \subseteq B$ such that $l\cap A_1\neq\emptyset,l\cap A_2\neq\emptyset$ and $l \subseteq \cup_{j=1}^{r}A_j$. However, since $A_1, \ldots, A_r$ are disjoint with pairwise distance of at least $2\omega$, there exists a line segment $l'\subset l$ of length $\omega$ such that $l'\not\subset\cup_{j=1}^{r}A_j$, contradicting with $l\subset B\subseteq\cup_{j=1}^{r}A_j$.
    
    Therefore, we can find a disc $C$ such that $C^{\circ} \subseteq B \setminus \cup_{j=1}^r {A_j}$. We further claim that there exist some $j,k\in[r]$ such that $C\cap A_j\neq \emptyset$ and $C\cap A_k\neq \emptyset$. Otherwise, we would have $C\cap A_i\neq \emptyset $ for at most one $i\in[r]$. In this case, we may perturb the centre of $C$ (away from the disc that it is currently touching, if any) while adjusting its radius, up until the point where $C\cap A_j\neq \emptyset $ for at least two $j\in[r]$.
    
    Finally, we establish that $C$ has a radius at least $\omega$. Let $x$ be the centre of $C$. We have 
    \[
    2\omega \leq \mathrm{dist}(A_j, A_k) \leq \mathrm{dist}(x, A_j) + \mathrm{dist}(x, A_k) \leq 2\mathrm{Rad}(C),
    \]
    as desired.
\end{proof}
\begin{lemma}\label{Lemma:ChompingDoughnut}Let $A,B\in\mathcal{S}$ be two $d$-dimensional discs such that $A\subset B$, and $d(A,B^{c})\geq\omega$ for some $\omega\in(0,\pi)$. For $r\geq 1$, let $R_1,\dots,R_r$ be $r$ discs on a $d$-dimensional Riemannian manifold $\mathcal{M}$ such that $R_i\cap B\neq\emptyset$, $R_i\cap A=\emptyset$, and $d(R_i,R_j)\geq\omega$ for $i,j\in[r]$ such that $i\neq j$. Then there exists a disc $D\subset\mathcal{M}$ of radius $\omega/8$ such that $D\subset B\setminus(\cup_{i=1}^{r}R_i\cup A)$.
\end{lemma}
\begin{proof}
    Denote $x:=\mathrm{Ctr}(B)$ and $r:=\mathrm{Rad}(B)$. Let $S:=\{y\in B:\mathrm{Geo}(x,y)\leq r-\omega/2\}$ be a disc centred at $x$ with radius $r-\omega/2$ so that $A\subset S\subset B$, and denote the boundary of $S$ as $\partial S$ which is a $(d-2)$-dimensional sphere.
    Let $C\subset B\setminus S$ be a $(d-1)$-dimensional disc with radius $\omega/4$. Such disc exists because $\mathrm{dist}(S,B^{c})=\omega/2$.
    
    We claim that there exists a disc $D\subset C$ such that $\mathrm{Rad}(D)=\omega/8$ and $D\cap(\cup_{i=1}^{r}R_i)=\emptyset$. To show this, we consider three cases regarding $C\cap (\cup_{i=1}^{r}R_i)$.
    
    Firstly, when $C\cap (\cup_{i=1}^{r}R_i)=\emptyset$, the claim holds by choosing $D=C$.
    
   Secondly, assume there exist $i,j\in[r]$ such that $C\cap R_i\neq\emptyset$ and $C\cap R_j\neq\emptyset$. By Lemma~\ref{Lemma:ChompingBiscuits}, there exists a disc $D$ of radius at least $\omega/8$ such that $D\subset C\setminus(\cup_{i=1}^{r}R_i)$. Since $C\subset B\setminus A$, we have $D\subset B\setminus (A\cup(\cup_{i=1}^{r}R_i))$ as desired. 
    
    Thirdly, assume $C\cap R_i\neq\emptyset$ for a unique $i\in[r]$. We have $(B\setminus S)\not\in R_i$ for any $i\in[r]$, for otherwise there exists $j\in[r]$ such that $B\setminus S\subseteq R_j$ violating the radius constraint that $\mathrm{Rad}(R_j)\leq\pi/2$. Then one may move $C$ inside $B\setminus S$ until the relationship between $C$ and $\cup_{i=1}^{r}R_i$ changes to either the first case where $C\cap(\cup_{i=1}^{r}R_i)=\emptyset$ holds or the second case where $|i\in[r]:C\cap R_i\neq\emptyset|\geq2$ holds.

    Combining the three cases, we have found a disc $D\subset B\setminus(\cup_{i=1}^{r}R_i\cup A)$ of radius $\omega/8$ and the proof is complete.
\end{proof}
\begin{lemma}\label{proof:lemma_exact_one_disc_covered}
	Let $R_1,\dots,R_r$ be $r$ discs on $\mathbb{S}^{d-1}$ satisfying $\mathrm{dist}(R_i,R_j)\geq\omega$ for some $\omega\in(0,\pi)$. Let  
	$B_1,\dots,B_{J}$ be independent and identically distributed random discs on $\mathbb{S}^{d-1}$ such that $\bigl(\mathrm{Ctr}(B_1), \mathrm{Rad}(B_1)\bigr) \sim \mathrm{Unif}(\mathbb{S}^{d-1})\otimes \mathrm{Unif}[0,\pi]$. Given some $\eta\in(0, \omega)$, define an event 
	\[
	\mathcal{E}_1 = \{ \forall i\in[r], \exists j \in [J] \text{ s.t. } R_i \subseteq B_j,\mathrm{dist}(R_i,B_j^{\mathrm{c}})\geq \eta, \cup_{k:k\neq i}R_k \cap B_j = \emptyset\}.
	\]
	Then
	\begin{align*}
		\mathbb{P}(\mathcal{E}_1)&
		\geq 1-r\exp(-C_d J (\omega-\eta)^d ),
	\end{align*}
	where $C_{d}=1/(2^dd^2\Gamma(d/2-1/2 ))$. 
\end{lemma}
\begin{proof}
	Fix $i\in[r]$ and $j\in[J]$. Let $a\in\mathbb{S}^{d-1}$ and $b\in [0,\pi]$ be the centre and radius of $R_i$ and $U$, $R$ be the centre and radius of $B_j$. We have
	\begin{align*}
		\mathbb{P}\bigl\{\mathrm{Nhd}(R_i,\eta) & \subseteq B_j\subseteq \mathrm{Nhd}(R_i,\omega)\bigr\} = \mathbb{P}\bigl\{\mathrm{dist}(U, a) \leq \min\{b + \omega - R, R - b - \eta \}\bigr\}\\
		&\geq \mathbb{P}(R \in [b+\eta, b+(\omega+\eta)/2], \; \mathrm{dist}(U, a)\leq R-b-\eta)\\
		&=\int_{0}^{(\omega-\eta)/2} \frac{\pi^{d/2-1/2}}{\Gamma \left( \frac{d-1}{2} \right)} \int_0^{\sin^2(x)}t^{\frac{d}{2}-\frac{3}{2}} (1-t)^{-\frac{1}{2}} dtdx\\
		&\geq\frac{2\pi^{d/2-1/2}}{(d-1)\Gamma(\frac{d-1}{2})}\int_{0}^{(\omega-\eta)/2}\sin^{d-1}(x)dx\\
        &\geq \frac{2\pi^{d/2-1/2}}{(d-1)\Gamma(\frac{d-1}{2})}\int_{0}^{(\omega-\eta)/2}(2x/\pi)^{d-1}dx\\
		&\geq C_{d}(\omega-\eta)^{d},
	\end{align*} 
	where $C_{d}=1/(2^dd^2\Gamma(d/2-1/2 ))$ and the penultimate inequality holds because $\sin(x)\geq 2x/\pi$ for $x\in [0,\pi/2]$.
	For a fixed $i\in\{1,\dots,r\}$, denote the event $\mathcal{E}_{1,i}=\{\exists j\in[J],\,\mathrm{Nhd}(R_i,\eta)\subseteq B_j\subseteq \mathrm{Nhd}(R_i,\omega)\}$, because $\{B_j\}_{j\in[J]}$ are i.i.d., we have
	\begin{align*}
		\mathbb{P}(\mathcal{E}_{1,i}^{\mathrm{c}})\leq (1-C_d(\omega-\eta)^{d})^{J}\leq \exp(-C_d(\omega-\eta)^{d}J),
	\end{align*}
	thus
	\begin{align*}
		\mathbb{P}(\mathcal{E}_1) &=\mathbb{P}(\cap_{i\in[r]}\mathcal{E}_{1,i})\geq 1-\sum_{i=1}^{r}\mathbb{P}(\mathcal{E}_{1,i}^{\mathrm{c}})\geq1-r\exp(-C_d(\omega-\eta)^{d}J),
	\end{align*}
	as desired.
	\end{proof}
\begin{lemma}\label{Lemma:CUSUM_signal}
 Let $\mathcal{D}=\{X_1,\ldots,X_n\}$ be deterministic design points on a manifold $\mathcal{M}$. For subsets $A, R_1\subseteq \mathcal{M}$, let $r:=|R_1|_{\mathcal{D}}$, and $\delta:=\min\{|A\triangle R_1|_{\mathcal{D}},|A\triangle R_1^c|_{\mathcal{D}}\}$. Suppose a vector $\mu = (\mu_1,\ldots,\mu_n)^\top$ satisfies $\mu_i = 0$ for $X_i\notin R_1$. If $\mu_i = \theta > 0$ for $X_i\in R_1$, then
	\renewcommand{\labelenumi}{(\alph{enumi})}
    \begin{enumerate}
		\item $\sqrt{\min\{r,n-r\}/ 2}\theta \leq \mathcal{T}_{R_1}(\mu) \leq \sqrt{\min\{r,n-r\}}\theta$;
        \item $\theta^2\min\{r,n-r,\delta\}/2 \leq \mathrm{RSS}_A(\mu) \leq \theta^2\min\{r,n-r,\delta\}$;
        \item $\mathcal{T}_{R_1}(\mu)\geq \mathcal{T}_{A}(\mu)$.
	\end{enumerate}
More generally, if $\mu_i$ are not necessarily equal for $X_i\in R_1$, writing  $\theta_{\min}=\min_{i:X_i\in R_1} |\mu_i|$ and $\theta_{\mathrm{rms}}=\{r^{-1}\sum_{i:X_i\in R_1}\mu_i^2\}^{1/2}$, then 
\renewcommand{\labelenumi}{(\alph{enumi})}
\begin{enumerate}
	\item[(d)] $\mathcal{T}_{A}(\mu)\leq \sqrt{r}\theta_{\mathrm{rms}}$;
	\item[(e)] 2$\mathrm{RSS}_{A}(\mu)\geq\min\{\theta^2_{\mathrm{rms}}r,\theta_{\min}^2(n-r),\theta_{\min}^2\delta\}$.
\end{enumerate}
\end{lemma}
\begin{proof} (a) From definition we have
	\begin{align*}
		\mathcal{T}_{R_1}(\mu) &= \sqrt{\frac{r(n-r)}{n} }\theta.
	\end{align*}
	The desired inequalities follows from the fact that $\min\{r,n-r\}/2 \leq r(n-r)/n \leq \min\{r,n-r\}$.

\bigskip 

	\noindent (b) Observe that
		\begin{align}\label{Eq:LemmaCUSUMSignalRSS}
		\mathrm{RSS}_{A}(\mu)&=\sum_{i=1}^{n}\mu_i^2-(\sum_{i:X_i\in A}\mu_i)^2/|A|_{\mathcal{D}}-(\sum_{i:X_i\in A^{\mathrm{c}}}\mu_i)^2/|A^{\mathrm{c}}|_{\mathcal{D}}\nonumber\\
        &=\theta^2\biggl(|R_1|_{\mathcal{D}}-\frac{|A\cap R_1|^2_{\mathcal{D}}}{|A|_{\mathcal{D}}}-\frac{|A^{\mathrm{c}}\cap R_1|^2_{\mathcal{D}}}{| A^{\mathrm{c}}|_{\mathcal{D}}}\biggr)\nonumber\\
		&= \theta^2\biggl\{\frac{|A\setminus R_1|_{\mathcal{D}}(|R_1|_{\mathcal{D}}-|R_1\setminus A|_{\mathcal{D}})}{(|R_1|_{\mathcal{D}}-|R_1\setminus A|_{\mathcal{D}})+|A\setminus R_1|_{\mathcal{D}}}+\frac{|R_1\setminus A|_{\mathcal{D}}(|R_1^{\mathrm{c}}|_{\mathcal{D}}-|A\setminus R_1|_{\mathcal{D}})}{(|R_1^{\mathrm{c}}|_{\mathcal{D}}-|A\setminus R_1|_{\mathcal{D}})+|R_1\setminus A|_{\mathcal{D}}}\biggr\}.
	\end{align}
From~\eqref{Eq:LemmaCUSUMSignalRSS} and the fact that $ab/(a+b)\geq\min\{a,b\}/2,\forall a,b\in\mathbb{R}$, we have
\begin{align*}
	\mathrm{RSS}_A(\mu)&\geq \frac{\theta^2}{2}\min\{|A\setminus R_1|_{\mathcal{D}},|R_1|_{\mathcal{D}}-|R_1\setminus A|_{\mathcal{D}}\}+\frac{\theta^2}{2}\min\{|R_1\setminus A|_{\mathcal{D}},|R_1^{\mathrm{c}}|_{\mathcal{D}}-|A\setminus R_1|_{\mathcal{D}}\}\\
	&\geq \frac{\theta^2}{2}\min\{|A\triangle R_1|_{\mathcal{D}},|(A\triangle R_1)^{\mathrm{c}}|_{\mathcal{D}},|R_1|_{\mathcal{D}},|R_1^{\mathrm{c}}|_{\mathcal{D}}\}= \frac{\theta^2}{2}\min\{\delta, r,n-r\},
\end{align*}
 where the second inequality follows from the fact that $\min\{a,b\}+\min\{c,d\}=\min\{a+c,a+d,b+c,b+d\},\forall a,b,c,d\in\mathbb{R}$. Combining~\eqref{Eq:LemmaCUSUMSignalRSS} and the fact that  $ab/(a+b)\leq \min\{a,b\},\forall a,b\in\mathbb{R}$, we have 
 \begin{align*}
 	\mathrm{RSS}_A(\mu)&\leq \theta^2\min\{|A\setminus R_1|_{\mathcal{D}},|R_1|_{\mathcal{D}}-|R_1\setminus A|_{\mathcal{D}}\}+\theta^2\min\{|R_1\setminus A|_{\mathcal{D}},|R_1^{\mathrm{c}}|_{\mathcal{D}}-|A\setminus R_1|_{\mathcal{D}}\}\\
 	&=\theta^2\min\{|A\triangle R_1|_{\mathcal{D}},|(A\triangle R_1)^{\mathrm{c}}|_{\mathcal{D}},|R_1|_{\mathcal{D}},|R_1^{\mathrm{c}}|_{\mathcal{D}}\}\leq \theta^2\min\{\delta,r,n-r\},
 \end{align*}
 as desired in (b).

 \bigskip
 
\noindent (c) This follows directly from~\citet[Lemma 2]{baranowski2019narrowest}.

\bigskip 

\noindent (d)  Define a random variable $W\sim\mathrm{Unif}([n])$. Let $L=\mathbbm{1}\{X_W\in A\}$, and define $E_1=\mathbb{E}(\mu_W\mid L=1)$, $E_0=\mathbb{E}(\mu_W\mid L=0)$, $p_0=\mathbb{P}(L=1)$ and $p_1=\mathbb{P}(L=0)$. Then $E=p_1E_1+p_0E_0 = \mathbb{E}(\mu_W)$. Observe that
 \begin{align*}
 	\mathcal{T}_{A}^2(\mu)=np_0p_1(E_1-E_0)^2=np_1(E_1-E)^2+np_0(E_0-E)^2=n\mathrm{Var}(\mathbb{E}(\mu_W\mid L)),
 \end{align*}
 and by law of total variance we have $\mathcal{T}^2_{A}(\mu)\leq n\mathrm{Var}(\mu_W)\leq n\mathbb{E}(\mu_W^2)=r\theta_{\mathrm{rms}}^2$ as desired.

\bigskip 

\noindent (e) Define $U\sim\mathrm{Unif}(u\in[n]:X_u\in A)$ and $V\sim \mathrm{Unif}(v\in[n]:X_v\notin A)$, and $\theta_1^2:=\mathbb{E}(\mu_W^2\mid X_W\in A\cap R_1),\theta_0^2:=\mathbb{E}(\mu_W^2\mid X_W\in R_1\setminus A)$. Observe that $\mathrm{RSS}_A(\mu)=|A|_{\mathcal{D}}\mathrm{Var}(\mu_U)+|A^{c}|_\mathcal{D}\mathrm{Var}(\mu_V)$. Let $\mathcal{E}:=\{X_U\in R_1\}$. Again, by the law of total variance, we have
	\begin{align}\label{Eq:LowerBndVarU}
		\mathrm{Var}(\mu_U)&=\mathbb{E}(\mathrm{Var}(\mu_U\mid \mathbbm{1}_{\mathcal{E}}))+\mathrm{Var}(\mathbb{E}(\mu_U\mid \mathbbm{1}_{\mathcal{E}}))\nonumber\\
        &= \mathrm{Var}(\mu_U\mid {\mathcal{E}})\mathbb{P}(\mathcal{E}) + \{\mathbb{E}(\mu_U\mid {\mathcal{E}})\}^2 \mathbb{P}(\mathcal{E})(1-\mathbb{P}(\mathcal{E}))\nonumber\\
        &\geq \mathbb{E}(\mu_U^2\mid \mathcal{E})\mathbb{P}(\mathcal{E})(1-\mathbb{P}(\mathcal{E})) \geq \theta_1^2\mathbb{P}(\mathcal{E})(1-\mathbb{P}(\mathcal{E})).
	\end{align}
	Similarly we have $\mathrm{Var}(\mu_V)\geq\theta_0^2\mathbb{P}(\mathcal{E}')\mathbb{P}(\mathcal{E}')(1-\mathbb{P}(\mathcal{E}'))$ where $\mathcal{E}':=\{X_V\in R_1\}$. Hence,
	\begin{align}
		\mathrm{RSS}_A(\mu)&\geq\theta_1^2|A|_{\mathcal{D}}\mathbb{P}(\mathcal{E})(1-\mathbb{P}(\mathcal{E}))+\theta_0^2|A^c|_{\mathcal{D}}\mathbb{P}(\mathcal{E}')(1-\mathbb{P}(\mathcal{E}')) \nonumber\\
        &\geq \frac{1}{2}\theta_1^2 |A|_{\mathcal{D}}\min\{\mathbb{P}(\mathcal{E}), 1-\mathbb{P}(\mathcal{E})\} +\frac{1}{2}\theta_0^2 |A^c|_{\mathcal{D}}\min\{\mathbb{P}(\mathcal{E}'), 1-\mathbb{P}(\mathcal{E}')\}\\
        &\geq\frac{1}{2} \min\{\theta_{\mathrm{rms}}^2r, \theta_{\min}^2\delta, \theta_{\min}^2(n-r)\}
	\end{align}
    as desired, where the last inequality follows from the facts that $\theta_{1}^2\geq\theta_{\min}^2,\theta_0^2\geq\theta_{\min}^2$ and $|A|_{\mathcal{D}}\theta_1^2+|A^{c}|_{\mathcal{D}}\theta_0^2=r\theta_{\mathrm{rms}}^2$.
\end{proof}
\section{Time complexity}\label{appC}
Theoretically, the time complexity of Algorithm~\ref{algo:multiple_cp_detection} is $O(n^2dJ_n)$, where drawing discs takes $O(dJ_n)$, forming pairs of inner and outer discs takes $O(n)$, and computing RSS and local CUSUM take $O(n)$ for each disc pair. In case where $J_n$ is increasing with $n$, the running time scales cubically with $n$.

The running times of our estimator for $d\in\{2,3,4\}$ and $n\in\{50,100,200,500,1000\}$ are given in Table~\ref{tab:time_ours}, and we set $J_n=1000$ for each setting. The running times of competitors when $d=3$ are shown in Table~\ref{tab:time_all}. 
\begin{table}[ht]
	\centering
	\begin{tabular}{llccc}
		\toprule
		& $d = 2$ & $d = 3$ & $d = 4$ \\
		\midrule
		$n = 50$  & $3.401$ & $2.069$ & $2.023$ \\
		$n = 100$  & $4.372$ & $4.33$ & $4.301$ \\
		$n = 200$  & $9.843$ & $9.814$ & $9.827$ \\
		$n = 500$  & $36.845$ & $36.823$ & $36.86$ \\
		$n = 1000$  & $127.172$ & $126.08$ & $126.729$ \\
		\bottomrule
	\end{tabular}
	\caption{Running times of our estimator ($J_n=1000$) in seconds}
	\label{tab:time_ours}
\end{table}

\begin{table}[ht]
	\centering
	\begin{tabular}{llccccc}
		\toprule
		& $n = 50$ & $n = 100$ & $n = 200$ & $n = 500$ & $n = 1000$ \\
		\midrule
		ours  & $2.069$ & $4.33$ & $9.814$ & $36.823$ & $126.08$ \\
		dalponte2016  & $0.409$
 & $0.392$ & $0.391$ & $0.408$ & $0.414$ \\
		silva2016  & $0.339$ & $0.363$ & $0.378$ & $0.36$ & $0.369$ \\
		li2012  & $0.274$ & $0.271$ & $0.291$ & $0.275$ & $0.283$ \\
		\bottomrule
	\end{tabular}
	\caption{Running times in seconds}
	\label{tab:time_all}
\end{table}
\end{appendix}

\bibliographystyle{abbrvnat}
\bibliography{main}
\end{document}